\documentclass{article}
\usepackage[utf8]{inputenc}
\usepackage{fullpage}
\usepackage{graphicx}
\usepackage{amsfonts,amsmath,amssymb}
\usepackage{algorithm,algorithmic}
\usepackage{hyperref}
\usepackage{xcolor}
\usepackage{bm}
\usepackage{amsthm}
\usepackage{enumitem}
\usepackage{thmtools}
\usepackage{thm-restate}

\newtheorem{theorem}{Theorem}
\newtheorem{corollary}[theorem]{Corollary}
\newtheorem{lemma}[theorem]{Lemma}
\newtheorem{remark}{Remark}

\newtheorem{definition}{Definition}
\newcommand{\clique}{\ensuremath{\mathsf{Clique~}}}
\newcommand{\cliquefull}{\ensuremath{\mathsf{Congested ~Clique~}}}
\newcommand{\fclique}{\ensuremath{\mathsf{Faulty~Clique~}}}

\newcommand{\congest}{\ensuremath{\mathsf{Congest~}}}

\usepackage{blkarray}

\newcommand{\f}{\varphi} 
\newcommand{\de}{\delta} 
\newcommand{\g}{\mathrm{Enc}} 
\newcommand{\bw}{c} 
 
\newcommand{\bin}{bin} 
\newcommand{\mult}{mult}
\newcommand{\remainF}{F_{remain}}

\title{Computing in a Faulty Congested Clique}

\author{
Keren Censor-Hillel\\
Technion\\ 
Israel\\
\textit{ckeren@cs.technion.ac.il}\\
	\and
Pedro Soto\\
Virginia Tech \\
USA\\
\textit{pedrosoto@vt.edu}\\
}

\date{}

\begin{document}
\maketitle

\begin{abstract}
We study a \textsf{Faulty Congested Clique} model, in which an adversary may fail nodes in the network throughout the computation. 
We show that any task of $O(n\log{n})$-bit input per node can be solved in roughly $n$ rounds, where $n$ is the size of the network. This nearly matches the linear upper bound on the complexity of the non-faulty \textsf{Congested Clique} model for such problems, by learning the entire input, and it holds in the faulty model even with a linear number of faults.
    
Our main contribution is that we establish that one can do much better by looking more closely at the computation. Given a deterministic algorithm $\mathcal{A}$ for the non-faulty \textsf{Congested Clique} model, we show how to transform it into an algorithm $\mathcal{A}'$ for the faulty model, with an overhead that could be as small as some logarithmic-in-$n$ factor, by considering refined complexity measures of $\mathcal{A}$. 

As an exemplifying application of our  approach, we show that the $O(n^{1/3})$-round complexity of semi-ring matrix multiplication [Censor{-}Hillel, Kaski, Korhonen, Lenzen, Paz, Suomela, PODC 2015] remains the same up to polylog factors in the faulty model, even if the adversary can fail $99\%$ of the nodes (or any other constant fraction).
\end{abstract}

\tableofcontents

\section{Introduction}
Distributed systems are prone to failures by their nature, and thus coping with faults in distributed computing has been extensively studied since the dawn of this research area \cite{AttiyaWelch, Lynch96, peleg2000distributed}.   
In this work, we address the \cliquefull model (or \clique for short) \cite{LotkerPPP05}, in which $n$ nodes of a network communicate in synchronous rounds by exchanging $O(\log{n})$-bit messages between every pair of nodes in each round. This model is heavily studied, through the lens of various computing tasks, such as algebraic computations \cite{Censor-HillelKK19, Gall16, Censor-HillelLT20}, graph problems such as computing an MST \cite{LotkerPPP05, HegemanPPSS15, GhaffariP16, Korhonen16, Jurdzinski018, Nowicki21a}, computing distances and spanners \cite{HenzingerKN16, Nanongkai14, DoryFKL21, Censor-HillelDK21, DoryP22, BuiCCDL24, ParterY18}, computing local tasks \cite{HegemanPS14, Ghaffari17, GhaffariGKMR18, Censor-HillelPS20, CzumajDP21, ParterS18, Parter18, CoyCDM23}, optimization and approximation algorithms \cite{HegemanPS14, GhaffariN18, GhaffariJN20}, subgraph finding \cite{DruckerKO13, DolevLP12, Pandurangan0S18, Censor-HillelFG22, IzumiG17, IzumiG19, FischerGKO18, Censor-HillelFG20, Censor-HillelGL20} and many more \cite{Patt-ShamirT11, Lenzen13, HegemanP15, Censor-HillelMP21} (see Section \ref{appendix:relaredWork} for additional related work). 

We consider deterministic algorithms, and study a faulty version of the \clique model, which we refer to as the \fclique model. In this model an adversary may fail nodes in the network throughout the computation, such that a failed node cannot continue to participate in the computation. We capture the budget of failures that the adversary has as a parameter $c$, such that at least $n/c$ nodes must remain non-faulty. Thus, the adversary may fail $(\frac{c-1}{c})n$ nodes throughout the execution of an algorithm (see Section \ref{sec:preliminaries} for a  formal definition of the model and the additional concepts used throughout the paper). We ask:

\begin{center}
\emph{How efficiently can one compute in the \fclique model?}
\end{center}

Our contribution is a scheme that converts any \clique algorithm into a \fclique algorithm. The goal of such a scheme is to incur the smallest possible overhead to the round complexity of the algorithm. We exemplify the strength of our scheme by showing that it allows multiplying two matrices over a semi-ring in $\tilde{O}(n^{1/3})$ rounds, thus retaining its non-faulty complexity from \cite{Censor-HillelKK19} up to polylog factors. This is \emph{more than quadratically faster} compared to the complexity of $n^{2/3+o(1)}$ that can be obtained for this problem from the follow-up work of \cite{CFGS25}.

\paragraph{Let's dive into some background.} Without further restrictions on the adversary, it may fail a node at the very first round of the algorithm, preventing the others from ever completing the computation since they can never access the failed node's input. Note that we do not allow loss of information since our requirement is that the output depends on the original inputs, as opposed to some other fault models that allow the output to depend only on the inputs of non-faulty nodes. 
Thus, the \fclique model must allow at least one initial \emph{quiet round}, in which no faults can occur. We therefore consider the number of quiet rounds that are needed for a \fclique algorithm as an additional complexity measure which we aim to minimize. We emphasize that our method will require a very small constant number of quiet rounds. This does not render the model trivial for any task which can be solved in $O(1)$ rounds in the non-faulty \clique model, because if its exact complexity is larger than our constant then it cannot be entirely executed during the quiet rounds (also note that we do not modify the bandwidth by any constant, but rather we stick to a certain given bandwidth of $b\log n$ bits, as we cannot change the communication model based on the task that we wish to solve).\footnote{\cite{CFGS25} work with already-coded inputs and thus avoid the notion of quiet rounds. Here, we care about how the input is coded and hence we need these constant number of preprocessing rounds.}

Similarly, the adversary can fail a node at the end of the computation, preventing its output from being accessible. Thus, in the \fclique model, the output requirement is modified. Yet, some care should be taken when doing so. We denote by $\mathcal{A}_w$ the output that node $w$ should hold in a non-faulty execution of $\mathcal{A}$, and we would like to require that for every node $w$ there is a node $u$ which holds $\mathcal{A}_w$. However, this is still insufficient because the adversary may now fail $u$. Instead, we demand that for every $w$, the output $\mathcal{A}_w$ is encoded in the network such that any node $u$ can obtain it within some number of $R$ rounds of communication, even if additional nodes fail, as long as $u$ itself does not fail. We call this the decodability complexity of the algorithm. Thus, a \fclique algorithm has three complexity measures: the number of quiet rounds, the number of additional (not necessarily quiet) rounds, and the number of rounds required for obtaining an output.

Our contribution is a method for computing in the \fclique model that yields the following. First, a simplified usage of it gives that any task that has $O(n\log{n})$ bits of input per node can be solved in $\tilde{O}(n)$ rounds in the \fclique model, for constant $c$ (our results hold for larger values of $c$ as well, with the actual dependence on $c$ being polynomial in $c$). This nearly matches the upper bound of linear round complexity of computing such tasks in the non-faulty \clique model, and is similarly obtained by learning all the inputs.

Second, our main result improves upon this greatly for computing certain functions $f$: To this end, we identify two particular parameters which, informally,  capture the locality of communication and the locality of computation that a certain computation requires, and show that our result obtains a super-fast computation of functions for which these parameters are not large (see Section \ref{sec:preliminaries} for formal definitions). For such functions, the complexity overhead of our approach can be as low as $\tilde{O}(1)$ for a constant $c$. This insight is what allows us to compute the semi-ring product of two matrices, whose complexity in the \clique model is $O(n^{1/3})$ due to \cite{Censor-HillelKK19}, in $\tilde{O}(n^{1/3})$-rounds in the \fclique model.


\subsection{Our Contribution}
\label{subsec:contribution}
Our treatment of \clique and \fclique algorithms goes through circuits, as was also done in the follow-up work of \cite{CFGS25}, and is aligned with prior work of \cite{DruckerKO13} which show how to compute circuits with \clique algorithms. For completeness and for setting the ground with the terminology that we need, we define in Section \ref{sec:preliminaries} the concept of \emph{layered circuits}. Roughly speaking, these are circuits whose gates can be nicely split into layers, with wires only between two consecutive layers. When we compute a circuit by a \clique or \fclique algorithm, we assign each node with the task of computing some of the gates in a layer, which we refer to as this node's \emph{part}. Note that a node always refers to a computing component in the \clique or \fclique model, while we use gates and wires to refer to circuits. 

Layered circuits can be parametrized by what we call their \emph{parallel partition parameter}. Informally, this parameter relates to the number of gates that have wires into and from the gates of any node's part.
This captures the amount of ``communication'' between layers when we later interpret them as \clique algorithms. Intuitively, we say that a layered circuit has a parallel partition of $n^{\f}$ when its gates have wires into/from $O(n^{1+\f})$ gates in other parts. Thus, when we compute a layer of the circuit by a \clique algorithm, the multiplicative overhead beyond the $n$ elements that each node can send/receive in a round is $O(n^{\f})$. This means that the \clique algorithm will incur this number of rounds per layer that it computes. It is not hard to show that any \clique algorithm that runs in $T$ rounds has a circuit of depth $O(T)$ with a 1-parallel partition ($\f=0$). Notably, this \emph{does not} depend on any non-adaptivity assumption on the \clique algorithm (that is, this applies also to algorithms in which the communication pattern may depend on the inputs). We show this in Lemma \ref{lem:clique_to_circ} in Section \ref{appendix:clique2circuit}.

Our key insight, is that for handling faults in the \fclique model, we can exploit additional parameters of a circuit that we wish to compute. To this end, we define two additional parameters, which we call \emph{communication locality} and \emph{computation locality}. These help us capture the amount of ``locality'' between layers when we later interpret them as \fclique algorithms. Informally, the computation locality of a layered circuit is $n^{\zeta}$ when in each layer, each of the $n$ parts can be split to $O(n^{\zeta})$ pieces of $n$ gates (for a total of at most $O(n^{1+\zeta})$ gates). The communication locality is $n^{\xi}$, when for each of the $n$ parts, the number of pieces of other parts that have incoming wires into it is at most $O(n^{\xi})$.

Intuitively, the {computation locality} measures the amount of information that a node needs to store in order to proceed with computing a layer of a circuit. The reason that this plays a role is because in the \fclique model we will need a non-faulty node $u$ to simulate the parts of the computation of a faulty node $w$. To do that, the non-faulty node $u$ will need to recover information that corresponds to the state of the faulty node $w$. However, it may be that $w$ does not need its  entire state, and that it can instead be compressed or partially ``forgotten''. This will allow a better complexity for the \fclique algorithm.
The {communication locality} captures the number of $\Theta(n)$-sized pieces of other nodes, which a non-faulty node $u$ has to get in order to simulate the computation of a faulty node $w$ in a layer. That is, our algorithm would benefit from a circuit in which if a node $u$ needs information from a node $w$ then it needs entire pieces of $n$ bits rather than smaller chunks of information. One can think of the information of $w$ as checkpointed ``in the cloud'', and a smaller communication locality will also reduce the complexity of the \fclique algorithm. 

The reason that communication locality only addresses incoming wires is that 
the parallel partition parameter $n^{\f}$ already takes care of communication regardless of faults. However, in the \fclique model we cannot trust the ability of a node to send its outgoing messages because it or other nodes may be faulty. Thus, it still remains to capture the number of coded pieces that a node has to collect in order to compute its gates in the upcoming layer.

We are now ready to state our main result (proven in Section \ref{sec:circuitTOfclique}). For simplicity, we assume throughout the paper that $c$ is a constant (our results apply to all values of $c$ with an additional dependence on $c$). We prove that every layered circuit can be computed by a \fclique algorithm, with $O(1)$ quiet rounds and decodability, and with a round complexity that depends on the above parameters, as follows.

\begin{restatable}[Computing a layered circuit by a \fclique algorithm]{theorem}{ThmCiruitToFclique}
\label{thm:circ_to_clique_coded}
    Let $\mathcal{C}$ be a circuit of depth $n^\de$ with alphabet size $|\Sigma| = 
    b \log(n)$
    that has an $n^{\f}$-parallel partition with computation locality $n^{\zeta}$ and communication locality $n^{\xi}$. Let $\mu$ be such that $O(n^{1+\mu})$ bounds the max size of the output per part.
    For every constant $\bw$, there is an algorithm $\mathcal{A}$ in the $\bw$-\fclique 
    that computes the function $f_\mathcal{C}$ with $O(1)$ quiet rounds, a (non-quiet) round complexity of $O(c^2(n^{\de+\f+\eta}+n^{\mu})\log{n})$, where $\eta=\max\{\zeta-\f, \xi-\f\}$, and decodability complexity of $R=O( n^{\mu})$.
\end{restatable}

The layered circuit we obtain in Lemma \ref{lem:clique_to_circ} for a general \clique algorithm with $T$ rounds has a depth of $O(T)$, which is $n^{\de}$ for $\de=\log_n{T}$, has an $n^\f$-parallel partition with $\f=0$, communication locality of $n^{\zeta}$ for $\zeta=\log_n{T}$, and computation locality of $n^{\xi}$ for $\xi=1$. Thus, Theorem \ref{thm:circ_to_clique_coded} gives us a $c$-\fclique algorithm with a round complexity of $O(Tn\log{n})$, which is a roughly linear-in-$n$ overhead for a constant $c$. This is expensive for a general function $f$, as one can use a trivial circuit that has depth 1 to get a $\tilde{O}(n)$ round complexity (in which one can verify that $\delta=0$,  $\mu=0$, $\varphi=1$, $\xi=0$, $\zeta=1$, and hence $\eta=0$), 
by essentially having each node learn all inputs. Yet, if there is a function $f$ for which one can construct a better circuit than the general construction of Lemma \ref{lem:clique_to_circ}, then one can get a smaller overhead compared with the \clique algorithm. Indeed, this is exactly what we exemplify for semi-ring matrix multiplication in Section \ref{appendix:semi-ringMM}.

\begin{restatable}[A \fclique algorithm for semi-ring matrix multiplication]{theorem}{ThmMatrix}
\label{thm:matrix}
Suppose that the inputs on the nodes are matrices 
$A, B \in \Sigma^{n\times n }$,
where each node begins with $n$ coefficients of $A$ and $B$ and $\Sigma $ is a semi-ring, then the value $C = A \cdot B$ can be computed in the \bw-\fclique with $O(1)$ quiet rounds, a round complexity of  $O(c^2 n^{1/3}\log n)$, and $O(1)$-decodability. 
\end{restatable}
In particular, Theorem \ref{thm:matrix} means that for any constant $c$, we incur only a logarithmic overhead in the round complexity over the non-faulty \clique model, in which this task can be solved in $O(n^{1/3})$ rounds due to \cite{Censor-HillelKK19}. Following up on our work, \cite{CFGS25} present a different scheme, which yields a complexity of $T^2\cdot n^{o(1)}$ rounds in the \fclique model for any $T$-round \clique algorithm. While for some \clique algorithms their approach improves upon ours, for others our approach is significantly faster: their scheme would give a complexity of $n^{2/3+o(1)}$ for semi-ring matrix multiplication, which is \emph{quadratically} slower compared to ours. We emphasize that matrix multiplication is just an example, and there are more cornerstone tasks that can enjoy our approach by constructing circuits that are efficient for the use of Theorem \ref{thm:circ_to_clique_coded}.

The way that the above is obtained is by compressing the communication rounds into a constant number of layers, while still having small computation and communication locality parameters, that are equal to the parallel partition parameter $n^{\f}$ (that is, $\zeta=\xi=\f$). Thus our algorithm incurs no additional overhead. We say that such a circuit has an $n^{\f}$-\emph{local} parallel partition. That is, since our Theorem \ref{thm:circ_to_clique_coded} shows that we can trade off $\de$ and $\f$ when constructing the circuit, we leverage this while exploiting the concurrency depending on the refined locality parameters.

This exemplifies an immediate corollary of Theorem \ref{thm:circ_to_clique_coded}, which is that circuits with small {local parallel partitions} are the optimal choice for invoking Theorem \ref{thm:circ_to_clique_coded}, in the sense that our transformation incurs the smallest overhead for them when implementing them in the \fclique model. This should be viewed through the lens of what we can say about the complexity of their \clique implementation, which is $O(n^{\de+\f})$ round, as we prove in Lemma~\ref{lem:circ_to_clique} in Section \ref{appendix:non-faulty}.   
Note that by the combination of Lemma \ref{lem:clique_to_circ} and Lemma~\ref{lem:circ_to_clique} it holds that layered circuits are equivalent to \clique algorithms.

~\\\textbf{A remark on quiet rounds.} Matrix multiplication is an excellent example for why we cannot simply replace quiet rounds by an assumption of coded inputs. If we do not carefully shuffle the input matrices before encoding them, then the nodes will be forced to download (decode) too much information compared to what they need, resulting in communication and computation localities that are extremely higher compared to those of the circuit that we construct for proving Theorem \ref{thm:matrix}. More explicitly, if we want a matrix multiplication algorithm to re-distribute the entries of the second input matrix among the nodes such that each node holds a column of the matrix rather than a row, or if we want each node to hold any other pattern of $n$ entries as we do in Section \ref{appendix:semi-ringMM}, 
then we need the information to already be encoded in this manner, as otherwise a node has to decode too much information in order to obtain its relevant entries. This could result in a huge number of rounds for a task that takes $O(1)$ rounds in the non-faulty \clique model.

~\\\textbf{Regimes of lower number of faults.} Finally, we show that if the bound on the number of failures is sublinear, then our method allows circuits with relaxed properties to be implemented by fast \fclique algorithms. The high level intuition is that if we are promised no more than $(c-1)n^{\chi}/c$ faults, we can replace the usage of codewords of length $n$ with $n^{1-\chi}$ disjoint codewords of length $n^{\chi}$. Since these are only technical modifications, we defer them to Section \ref{sec:sublinear}. This allows us to further extend our results for fast (ring) matrix multiplication, for which we show the following (see Section \ref{sec:fastMM}).

\begin{restatable}[A \fclique algorithm for Fast (Ring) Matrix Multiplication ]{theorem}{ThmFastMatrix}
\label{thm:ringmatrix}
Suppose that the inputs on the nodes are matrices 
$A, B \in \Sigma^{n\times n }$,
where each node begins with $n$ coefficients of $A$ and $B$ and $\Sigma $ is a ring, then the value $C = A \cdot B$ can be computed in the ($\chi,\bw$)-\fclique with $O(1)$ quiet rounds, a round complexity of  $O(\max\{c^2 n^{\chi }\log n,c^2 n^{1 - 2/ \omega}\log n \})$, and $O(1)$-decodability. 
\end{restatable}

Here $\omega $ is the exponent of matrix multiplication and $O(n^{1 - 2/ \omega})$ is the complexity of ring matrix multiplication in the non-faulty setting \cite{Censor-HillelGH19}.

\subsection{Technical Overview}
To prove Theorem \ref{thm:circ_to_clique_coded}, our approach is as follows. 
First, we partition the gates of each layer of the circuit into its respective parts in the parallel partition, and assign one node to each part. In the first $O(1)$ quiet rounds, we shuffle the inputs of the nodes such that each node holds the input wires to its gates, using Lenzen's routing scheme \cite{Lenzen13}. Then, each node $w$ encodes its state to prepare for the possibility of a failure, and sends one piece of the codeword to every other node. The code that is used has to be able to tolerate $(\frac{c-1}{c})n$ erasures, as this is the number of nodes that may fail (a node that has already failed simply does not receive its piece from $w$). We use $[n,n/c,(\frac{c-1}{c})n+1]_q$ codes, where $q$ is a field whose elements are of $\Theta(cb\log{n})$ bits, and thus we need $c=O(1)$ additional quiet rounds for this step. 

We then split the layers of the circuit into \emph{epochs}, where each epoch consists of computation layers and ends with a communication layer. On a high level, we say that a layer is a computation layer if for each part in the layer (a part is a set of gates assigned to a node for computing them), all of its output wires go in the next layer to the part that is assigned to the same node. This means that a node $w$ computing its gates in the next layer already has all the information it needs for the computation without needing to communicate with other nodes. A communication layer is any other layer. By the way we split into epochs, we get that once a node $w$ receives the information it needs to compute the gates in its part in the first layer of the epoch, it can compute the gates in its parts in the rest of the layers of the epoch. Eventually, we would want to have a circuit with as few epochs as possible. We emphasize that an algorithm may be represented by many different circuits and, indeed, given an algorithm, one needs to design a circuit that gives the parameters for which our result yields the best complexity measures.

In a non-faulty \clique algorithm, we would simply precede the above with communication between the nodes in order for $w$ to obtain the information it needs for computing the gates in its part in the first layer of the epoch (which is indeed what we do in Lemma \ref{lem:circ_to_clique}). However, in a \fclique algorithm, it is possible for the node $w$ to fail, in which case some node $u$ takes its role in computing the gates of $w$ in each layer in the epoch. Thus, each node $w$ encodes its state at the end of an epoch and sends one piece of the codeword to every other node. In the same spirit, at the beginning of the epoch, a node $u$ that simulates $w$ needs to first collect the pieces of the codeword. Communicating in order to obtain the required information for the epoch takes $O(cn^{\xi})$ rounds, since there are $n^{\xi}$ codes that need to be recovered, due to $n^{\xi}$ being the communication locality of the circuit. Encoding the information at the end of the epoch takes $O(cn^{\zeta})$ rounds, since there are $n^{\zeta}$ codes that need to be used, due to $n^{\zeta}$ being the computation locality of the circuit. Thus, this approach takes $O(c(n^{\zeta}+n^{\xi}))$ rounds per epoch. If every part of the last layer gets successfully encoded, this also implies decodability within a number of rounds that equals its output divided by $n$, as $n$ data elements can be routed to/from each node in a single round (again due to \cite{Lenzen13}).

However, the above is insufficient, because after a node $w$ fails, it may be that another node $u$ which simulates $w$ also fails. If we proceed with simulating $w$ by yet another node $u'$, we may incur a round complexity of the number of failures -- an unacceptable $(\frac{c-1}{c})n$ rounds. 

Thus, we simulate nodes in a more careful manner. For each epoch, we progress by \emph{attempts} for simulating the failed nodes of the epoch. In a certain attempt, if the number of simulation tasks that still need to occur for an epoch is greater or equal to the number of currently non-faulty nodes, then we let each non-faulty node simulate one of them. This promises that the attempt successfully simulates at least $n/c$ failed nodes, as this number of nodes is guaranteed to remain non-faulty. This type of attempt can only occur $O(c)$ times. 

The other case is when the number of remaining failed nodes that need to be simulated drops below the number of currently non-faulty nodes. In this case we simulate every such node by a multiplicity of non-faulty nodes, which implies that although we cannot argue that $n/c$ such tasks succeed, we have that it is still hard for the adversary to prevent such a task from succeeding because it would have to fail all the nodes that are assigned to it. We show that we make progress by succeeding in at least a constant fraction of the remaining tasks, resulting in at most a logarithmic number of such attempts. To make our analysis go through, we actually need to \emph{batch} these remaining tasks into batches of size $3c$, each of which is handled by some carefully chosen multiplicity of non-faulty nodes. 

Thus, one of the $c$ factors in our complexity is due to coding, and the other is due to either batching into $O(c)$-sized batches or requiring $c$ attempts per epoch (these latter two are disjoint events).
Finally, we incur an additive $O(c^2 n^{\mu} \log n )$ for checkpointing the outputs. Note that if $\mu$ is larger than the computation locality of the last epoch, then we can replace the computation of the last epoch and instead consider it as part of the decodability complexity: simply decode the checkpointed values from the penultimate epoch and compute the last epoch (that is, simulate the last computation epoch as part of the decodability).

We mention that, informally speaking, coding is essential for \fclique algorithms. We show in Lemma \ref{lem:must_code} that if the messages of the quiet rounds depend only on the input of each node and not on messages it receives, then the messages of $c$ quiet rounds must form codewords of an $[n,n/c,(\frac{c-1}{c})n+1]_q$ error-correcting code.

\subsection{Additional Related Work}
\label{appendix:relaredWork}
Already in the early 1980's, Pease et al. and Lamport et al.\cite{PeaseSL80, LamportSP82} introduced Byzantine Agreement, which requires $n$ parties to agree on one of their inputs, despite malicious behavior of some of the parties. Coping with Byzantine failures in our setting is an intriguing question. It is important to note that the faults we consider are fail-stop faults and not Byzantine faults. For the latter, there are provable limits of $1/3$ to the fraction of Byzantine parties \cite{PeaseSL80,FischerLM86}. In contrast, we can tolerate arbitrarily large amounts of crash failures.
We further mention that consensus algorithms for crash failures have been widely studied (see, e.g., \cite{AttiyaWelch}), but the known algorithms do not solve the general problem we consider here of computing \emph{arbitrary} circuits with \emph{small messages}.

Much earlier, von Neumann's \cite{von1956probabilistic} studied faulty gates in circuits. This differs from our setting of faulty nodes in the \fclique model, although we use circuits as a tool to describe their computations.
Also, \cite{Spielman96} started a line of work about the coded computing model for fault tolerant computations by multi-processors, in which the input and outputs of the system are coded. Our work for the \fclique model uses a similar approach of coding and decoding, but also takes into account the complexity of coding and decoding and aims to minimize them.

It is important to mention that error-correcting codes have been used in various types of distributed algorithms, this dates back to information dispersal of \cite{Rabin89} and goes through more recent results such as distributed zero knowledge proofs \cite{BickKO22}.

The literature about various types of fault tolerant distributed computations in various settings is too vast to survey here. Just to give a flavor, we mention a couple of lines of work about distributed graph algorithms prone to failures: In the \congest model, in which the communication network is an arbitrary graph and messages are of small $O(\log{n})$-bit size, various adversaries have been studied, which can cause edge or nodes failures, as well as Byzantine faults \cite{FischerP23,HitronPY22, HitronP21,HitronP21a}. In networks with unbounded amount of noise, \cite{Censor-HillelGH19, CensorHillelCGS23, FreiGGN24} show algorithms that can succeed in various computations despite having to cope with basically all message content being lost, or with asynchrony. In \cite{KumarMS22}, a slightly different faulty \cliquefull setting is considered for the specific task of graph realizations.

\section{Preliminaries}
\label{sec:preliminaries}
\subsection{The Model}
\begin{definition}[\textbf{The \clique model}]
\label{def:clique}
The \clique model is a synchronous communication model on $n$ nodes, where in each round every pair of nodes can exchange $b\log{n}$-bit messages, for some constant $b$. In an algorithm $\mathcal{A}$, each node has an input in $\Sigma^n$ for an alphabet $\Sigma$ of size $2^{b\log{n}}$, and after executing $\mathcal{A}$ each node holds an output in $\Sigma^M$ for some value of $M$.
\end{definition}

\begin{definition}[\textbf{The \fclique model}]
\label{def:faulty-clique}
The $\bw$-\fclique model is similar to the \clique model, except that an adversary may fail up to a fraction of $(\bw-1)/\bw$ of the nodes throughout the execution of an algorithm (that is, $n/c$ nodes must be non-faulty). A node that is failed in a certain round, does not send or receive any message from that round on. The \fclique model allows an initial number of \emph{quiet rounds}, in which no faults can occur. 
\end{definition}

\begin{definition}[\textbf{The complexity of a \fclique algorithm}]
\label{def:fclique-algorithm}
An algorithm $\mathcal{A}$ for the \fclique model has three complexity parameters. 

The first complexity measure is the number of quiet rounds it requires. The other two complexity measures are the round complexity and decodability complexity, defined as follows.

Let $\mathcal{A}$ be an algorithm for the \fclique model, and consider a non-faulty execution $\mathcal{E}_\textrm{non-faulty}$ of $\mathcal{A}$. Denote by $\sigma_{w,T_\textrm{non-faulty}}$ the state that node $w$ has at the end of round $T_\textrm{non-faulty}$ in $\mathcal{E}_\textrm{non-faulty}$.

For a round $T_\textrm{non-faulty}$ of $\mathcal{E}_\textrm{non-faulty}$, we say that a possibly faulty execution $\mathcal{E}$ of $\mathcal{A}$ satisfies at 
round $T$ the \emph{$R$-decodability condition} if the following holds. Suppose every node $u$ is associated with some other node $w_u$, then it is possible in $R$ rounds for each node $u$ that is non-faulty at the end of these $R$ rounds to obtain $\sigma_{w_u,T_\textrm{non-faulty}}$.
(Informally, the only way the adversary can prevent $u$ from receiving knowledge of \emph{any} other node's state after the next $R$ rounds is by failing $u$ itself.)

We stress that during these additional rounds the adversary may continue to fail nodes up to its given budget.
We say that the execution $\mathcal{E}$ is finished executing $\mathcal{A}$ in $T$ rounds with $R$-decodability, if at the end of $T$ rounds, it satisfies the $R$-decodability condition for the last round $T_\textrm{non-faulty}$ of $\mathcal{E}_\textrm{non-faulty}$ (note that it is possible that for other rounds of $\mathcal{E}_\textrm{non-faulty}$, the decodability condition in $\mathcal{E}$ holds with larger values than $R$).

The maximum values of $T$ and $R$ over all executions are the \emph{round complexity} and the \emph{decodability complexity} of $\mathcal{A}$, respectively.
\end{definition}

\subsection{Layered Circuits}
We define layer circuits and pinpoint which of their parameters captures the complexities of \clique and \fclique algorithms that compute them.

\begin{definition}[\textbf{fan-in, fan-out}]
\label{def:in}
For a vertex $v$ in a directed graph $G$, let 
$\mathrm{in}(v):=\{w \in V \mid (w,v) \in E \}$ 
and $\mathrm{out}(v):=\{w \in V \mid (v,w) \in E \}$.
The values $|\mathrm{in}(v)|$ and $|\mathrm{out}(v)|$ are called the \emph{fan-in} and \emph{fan-out} of $v$, respectively.
\end{definition}

\begin{definition}[\textbf{layered circuit}]\label{def:lay_circ}
A {layered circuit} $\mathcal{C}_n = \mathcal{C} = (V,E)$ of depth $n^\de$ over an alphabet $\Sigma$ is a connected directed graph with vertex set $V = V_0 \cup V_1 \cup ... \cup V_{n^\de}$ and edge set $E \subset \cup _{i \in [n^\de]} V_i \times V_{i+1}$, where each vertex (gate) $w$ with fan-in $F$ is labeled by a function $f_w:\Sigma ^{F} \rightarrow \Sigma $.

We say that a layered circuit $\mathcal{C}$ with $V_0=\{v^{(0)}_i\}_{i\in [k_{in}]}$ and $V_{n^\de}=\{v^{(n^\de)}_i\}_{i\in [k_{out}]}$ computes a function $f:\Sigma^{k_{in}}\rightarrow\Sigma^{k_{out}}$ if the following recursively defined function $f^\mathcal{C}$ is equal to $f$.
The function $f^\mathcal{C}$ is defined as the function whose output is $(f^\mathcal{C}_{v^{(n^\de)}_1},...,f^\mathcal{C}_{v^{(n^\de)}_{k_{out}}})$  where 
 \begin{enumerate}
     \item For every $1\leq i\leq k_{in}$, 
     $f_{v_i^{(0)}}:\Sigma \rightarrow\Sigma $ is the identity function and $f^\mathcal{C}_{v_i^{(0)}} = f_{v_i^{(0)}}(x_i) = x_i$, 
     where $x_i$ is the $i$-th input of $f$.

     \item  For every $1\leq\ell\leq n^\de$ and every $i$ such that $v_{i}^{(\ell)} \in V_\ell$ with fan-in ${F_{\ell,i}}$, there are gates $v^{\ell-1}_{i,j}$ for $j\in [{F_{\ell,i}}]$, such that     
     $f^\mathcal{C}_{v^{(\ell)}_i} = f_{v^{(\ell)}_i}(f^\mathcal{C}_{v^{(\ell-1)}_{i,1}},...,f^\mathcal{C}_{v^{(\ell-1)}_{i,{F_{\ell,i}}}})$.
 \end{enumerate}
\end{definition}

We make use of the following parameter of layered circuits, which we call the \emph{parallel partition parameter}, which captures the amount of ``communication'' between nodes for computing a layer of the circuit in a  \clique algorithm. 

\begin{definition}[\textbf{parallel partition}]
\label{def:par_part_graph}
   Let ${G}_n={G}=(V,E)$ be a 
   bipartite directed graph
   with $V = L \cup R$ and $E \subset L \times R$. For a pair of partitions $L = L_0 \cup L_1 \cup ... \cup L_{n-1}$ and $R = R_0 \cup R_1 \cup ... \cup R_{n-1}$, let 
   \begin{equation*}
       E_{i}^L := \{(\ell,r) \mid \ell\in L_i \text{ and } \exists j\neq i, r \in R_j\}, \quad
       E_{j}^R := \{(\ell,r) \mid r \in R_j \text{ and } \exists i\neq j, \ell \in L_i\}.
   \end{equation*}
   We say that ${G}$ has a {parallel partition} of size $n$ with {block fan size} $n^{1+\f}$ if there exists such a pair of partitions in which for every $i,j\in [n]$, $|E_i^{L}|,  |E_j^{R}| ={O(n^{1+\f})}$.
       We say that a layered circuit $\mathcal{C} = (V,E)$ of depth $n^\de$  has an $n^{\f}$-{parallel partition}, if there is a refinement $\mathcal{P}$ of the partition of $V = V_0 \cup V_1 \cup ... \cup V_{n^\de}$
     such that for every $i$, the restriction of 
    $ { \mathcal{P}}$ to each of the subgraphs $V_i \cup V_{i+1}$ is
    a parallel partition of size $n$ with block fan size {$n^{1+\f}$} with $L= V_i$, $R= V_{i+1}$. 
   \end{definition}

\sloppy The following parameters of layered circuits, which we call \emph{communication locality} and \emph{computation locality}, help us capture the amount of ``locality'' between layers when we later interpret them as \fclique algorithms.

\begin{definition}[\textbf{Computation locality and  Communication locality}]\label{def:locality}
Let $\mathcal{C} = (V,E) = (\cup_iV_i,E)$ be a layered circuit of depth $n^{\de}$ and an $n^\f$-{parallel partition} with respect to a refinement $\mathcal{P}$ of $V$. Let $P_{i,j}$ denote the $j$-th part of $V_i$ in $\mathcal{P}$.

We say that $\mathcal{C}$ has computation locality $n^\zeta$ and communication locality $n^{\xi}$ 
if there is a constant $h$ such that for all $P_{i,w}$ 
there exists a further (not necessarily disjoint) subdivision 
$P_{i,w} =: \bigcup_{j \in [hn^{\zeta}]}P^{(j)}_{i,w}$
such that 
$|P^{(j)}_{i,w}| = n$.
For each $u$, 
the number of pairs $j,w$ such that $w\neq u$ for which $P^{(j)}_{i,w}$ has a wire into $P_{i+1,u}$ is at most 
$h'n^\xi$ for some constant $h'$. We denote by $\bin(P_{i+1,u})$ the parts corresponding to those pairs, that is, the parts $P_{i,w}^{(j)}$ (where $w\neq u$)
in $V_{i}$ 
that consist of at least one gate that has a wire into a gate in the part $P_{i+1,u}$.

\end{definition}
\begin{definition}[\textbf{Local parallel partition}] \label{def:lpp}
    We say that a layered circuit $\mathcal{C} = (V,E) = (\cup_iV_i,E)$ of depth $n^{\de}$ has an $n^\f$-{local parallel partition}, 
    if it has an $n^\f$-parallel partition 
    for a refinement $V = \cup_{P \in \mathcal{P}}P$, with a computation locality of $O(n^{\f})$ and a communication locality of $O(n^\f)$.
\end{definition}

\section{\clique Algorithms as Layered Circuits}
\label{appendix:clique2circuit}
We prove that every \clique algorithm can be expressed by a layered circuit. It may be helpful for the reader to go through the construction of the circuit, to get a feeling of how a layered circuit may look like. However, a more informed reader may skip the proof and move on directly to our main result in Section \ref{sec:circuitTOfclique}.
\begin{lemma}\label{lem:clique_to_circ}
    Let $\mathcal{A}$ be a \clique algorithm that computes a function 
    $f$ in $T$ rounds. There is a circuit $\mathcal{C}$ of depth { $2T+1$} with alphabet size $|\Sigma| = b\log(n)$ that has a $1$-parallel partition (that is, $n^{\f}$ with $\f=0$), for which $f^\mathcal{C}=f$. 
\end{lemma}

\begin{proof}
~\textbf{Circuit construction:}
We construct the circuit $\mathcal{C}$ by first constructing the base layers $V_{0},V_{1},V_{2},V_{3}$, { corresponding to the first round}.
These are the cases $\ell = 0,1,2,3$; we then later consider the general cases $\ell = 2k, 2k+1$. See Figure \ref{fig:clique2circuit} for an illustration.

 ($\bm{\ell = 0}$)    We first construct $V_{0}:= \{v^{(0)}_{i,j}\}_{\substack{i\in [n],j\in [n]}}$ so that it represents the initial inputs of the nodes. Formally, for every $i\in [n]$ we denote $P_{0,i}:=\{v^{(0)}_{i,j}\}_{j\in [n]}$ as the set of gates that hold the $n$ inputs for node $i$, which we call the $n$ \emph{input gates} in $P_{0,i}$. Each such gate has fan-in 1 and fan-out $n+1$, which delivers its input wire as is to $n+1$ gates of $P_{1,i}$ in the next layer. 
    
  ($\bm{\ell = 1}$)  The layer $V_{1}:= \{v^{(1)}_{i,j}\}_{\substack{i\in [n],j\in [2n]}}$ represents the messages sent in the first round of $\mathcal{A}$. Each node $i\in [n]$ has \emph{outbox message gates}, defined as $\{v^{(1)}_{i,j}\}_{\substack{i\in [n],j\in [n]}}$, each has fan-in $n$, where for every $j\in [n]$, the gate $v^{(1)}_{i,j}$ receives as input the outputs of the all gates $v^{(0)}_{i,j}$ (for $j\in [n]$).
    The output of $v^{(1)}_{i,j}$, for $j \in [n]$, is the message that node $i$ sends to node $j$ on the first round.
    Each node also has its \emph{storage gates}, $\{v^{(1)}_{i,j}\}_{\substack{i\in [n],j\in [n,2n]}}$, each with fan-in of one wire, $e_{i,n+j}^{(0)} :=  v^{(0)}_{i,j} \rightarrow v^{(1)}_{i,n+j}$, for $j \in [n]$, 
    which stores all of its inputs.
    We denote 
    $P_{1,i}:=\{v^{(1)}_{i,j}\}_{j\in [2n]}$.

    ($\bm{\ell = 2}$)
    We construct the layer
    $V_{2}:= \{v^{(2)}_{i,j}\}_{\substack{i\in [n],j\in [2n]}}$ so that it represents
    the information needed to compute  either the messages that the nodes send in the second round or the outputs of $\mathcal{A}$, corresponding to whether $\mathcal{A}$ has $T>1$ or $T=1$ rounds, respectively. 
    Each node has {\emph{inbox message gates}} $v^{(2)}_{i,j}$, for $j \in [n]$, which represents the message that node $i$ receives from node $j$ on the first round and thus we have wires $e_{j,i}^{(1)} :=  v^{(1)}_{j,i} \rightarrow v^{(2)}_{i,j}$, for all $i,j \in [n]$.
Each node also has its storage gates, $\{v^{(1)}_{i,j}\}_{\substack{i\in [n],j\in [n,2n]}}$, each with fan-in of one wire, 
$e_{i,n+j}^{(1)} :=  v^{(1)}_{i,j} \rightarrow v^{(2)}_{i,n+j}$, which store all of its inputs. 
    We denote 
    $P_{2,i}:=\{v^{(2)}_{i,j}\}_{j\in [2n]}$.

($\bm{\ell = 3}$)
    We construct the layer
    $V_{3}:= \{v^{(3)}_{i,j}\}_{\substack{i\in [n],j\in [3n]}}$ so that it represents
    either the messages that the nodes send in the second round or the outputs of $\mathcal{A}$, corresponding to whether $\mathcal{A}$ has $T>1$ or $T=1$ rounds, respectively. 
    For $j \in [n]$, the inputs to $v^{(3)}_{i,j}$ are the wires $e_{i,j,k}^{(2)} :=  v^{(2)}_{i,k} \rightarrow v^{(3)}_{i,j}$, for all $k \in [2n]$; in particular, we let every such gate be connected to all of the  storage gates of that node in the previous layer. 
    If $T>1$, we define the output of $v^{(3)}_{i,j}$, for $j \in [n]$, to be the message that node $i$ sends to node $j$ on the second round.   
    If $T=1$, we define the output of $v^{(3)}_{i,j}$, for $j \in [Output]$, to be the output of $\mathcal{A}$, where the value of $Output$ is the size of output of each node. 
Each node also has its storage gates, $\{v^{(3)}_{i,j}\}_{\substack{i\in [n],j\in [n,3n]}}$, which each have fan-in of one wire, 
$e_{i,n+j}^{(2)} :=  v^{(2)}_{i,j} \rightarrow v^{(3)}_{i,n+j}$, which stores all of the values from the previous layer.
    We denote 
    $P_{3,i}:=\{v^{(3)}_{i,j}\}_{j\in [3n]}$.

($\bm{\ell = 2\ell'}$)
In general,  
for $\ell = 2\ell' $, we construct $       V_{\ell}:= \{v^{(\ell)}_{i,j}\}_{\substack{i\in [n],j\in [(\ell'+1)n]}}$ as follows. 
An inbox message gate $v^{(\ell)}_{i,j}$ for $j \in [n]$ has an incoming wire $e^{(\ell-1)}_{j,i}:=v^{(\ell-1)}_{j,i}\rightarrow v^{(\ell)}_{i,j}$. 
The storage gates $v^{(\ell)}_{i,j}$ for $j \in [n,(\ell'+1)n]$ store all of the data received by node $i$ up to round $\ell$. That is, each such gate has fan-in of one wire, 
$e_{i,n+j}^{(\ell-1)} :=  v^{(\ell-1)}_{i,j} \rightarrow v^{(\ell)}_{i,n+j}$. 
    We denote $P_{\ell,i}=\{v^{(\ell)}_{i,j}\}_{j\in [(\ell'+1)n]}$.

($\bm{\ell = 2\ell'+1}$)
If $\ell = 2\ell' +1 $ 
and $\ell < 2T+1$, then we construct $       V_{\ell}:= \{v^{(\ell)}_{i,j}\}_{\substack{i\in [n],j\in [(\ell'+2)n]}}$ such that an outbox message gate $v^{(\ell)}_{i,j}$ for $j \in [n]$ has an outgoing wire to $v^{(\ell+1)}_{j,i}$ 
and an incoming wire from 
all of the storage
gates from the previous round, i.e., $e_{i,j,k}^{(\ell-1)} :=  v^{(\ell-1)}_{i,k} \rightarrow v^{(\ell)}_{i,j}$, for all $k \in [(\ell'+1)n]$. We further have storage gates $v^{(\ell)}_{i,j}$ for $j \in [n,(\ell'+2)n]$ that store the all of the data received up to round $\ell$, i.e., each such gate has fan-in 1 with the wire $e_{i,n+j}^{(\ell-1)} :=  v^{(\ell-1)}_{i,j} \rightarrow v^{(\ell)}_{i,n+j}$.
Otherwise, if instead $\ell = 2T+1$, we define the output of $v^{(\ell)}_{i,j}$, for $j \in [Output]$, to be the output of $\mathcal{A}$, where the value of $Output$ is the size of output of each node. 
    We denote $P_{\ell,i}=\{v^{(\ell)}_{i,j}\}_{j\in [(\ell'+2)n]}$.

\begin{figure}[htbp]
            \centering
            \includegraphics[trim={1cm 5.5cm 2cm 9cm},clip,width=1\linewidth]{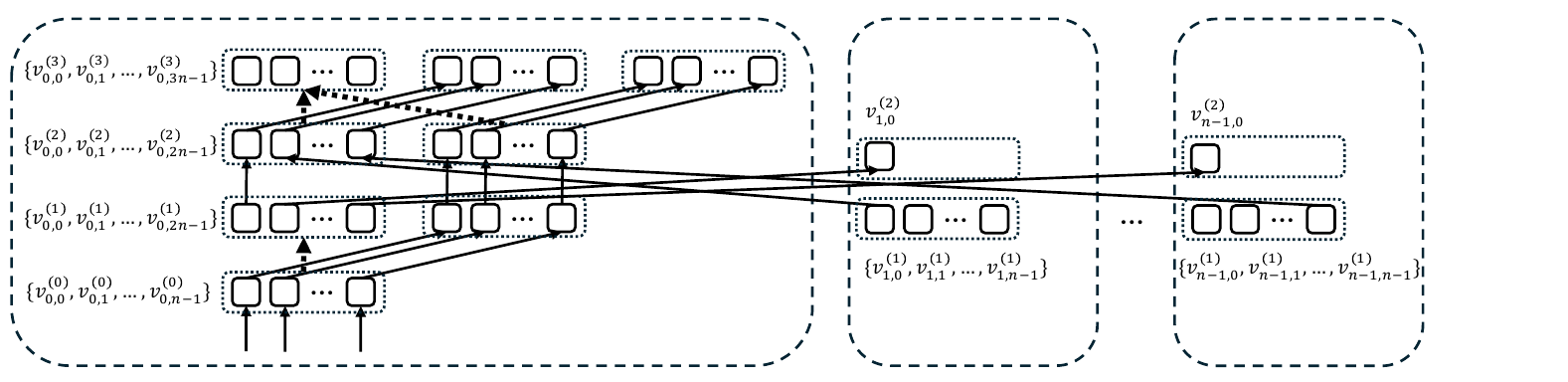}
            \caption{Illustration for layers $\ell=0,\dots,3$. The gates of node $0$ appear on the leftmost box, while for nodes $1$ and $n-1$ only several of their gates are depicted in the other two boxes. The solid thin arrows represent wires (some wires are omitted for clarity). A dashed wide arrow between two blocks of gates means that every gate in the block at the lower level has a wire into every gate in the block of the upper level.}
            \label{fig:clique2circuit}
        \end{figure}
    
~\\\textbf{Circuit parameters:} By construction, the depth of the circuit is thus {$2T+1$}. In addition, every consecutive pair of gate layers has a parallel partition with block fan size of 1. To see this, consider $V_{\ell},V_{\ell+1}$. If $\ell=2\ell'+1$ then for every $i\in [n]$, the only wires into gates of $P_{\ell+1,i}$ are from gates in $P_{\ell,i}$ and thus the block fan size is 0. Otherwise, if $\ell=2\ell'$ then for every $i\in [n]$, the only wires into $P_{\ell+1,i}$ which are from gates in $P_{\ell,i'}$ for $i'\neq i$ are the wires $e^{(\ell-1)}_{j,i}:=v^{(\ell-1)}_{j,i}\rightarrow v^{(\ell)}_{i,j}$ for $j \in [n]$. Thus, each block has one incoming wire from every other block, so the block fan size is $n$ (i.e., $\f=0$).

~\\\textbf{Correctness:}   
It remains to prove that $f_{\mathcal{C}}=f$. We prove the following by induction on the number of layers: If $\ell=2\ell'$ then for every $i,j\in [(\ell'+1)n]$, the outputs of the gates $v^{(2\ell')}_{i,j}$ are the inputs of node $i$ and all of the messages it receives in rounds 1 to $\ell'$ in $\mathcal{A}$. If $\ell=2\ell'+1$ then for every $i,j\in [n]$, the outputs of the gates $v^{(2\ell'+1)}_{i,j}$ are either the message that node $i$ sends in round $\ell'+1$ of $\mathcal{A}$ to node $j$ if $\ell' < T$ or the $j$-th output of node $i$ if $\ell'=T$, and the gates $v^{(2\ell'+1)}_{i,j}$ for $j\in [n,(\ell'+2)n]$ hold the inputs of node $i$ and all of the messages it receives in rounds 1 to $\ell'$ in $\mathcal{A}$.

\textbf{Base Step ($\ell=0,\dots,3$):} 
By construction, for $j\in [n]$, the gates $v^{(0)}_{i,j}$ receive the inputs that node $i$ has for $f$, and these are copied to their output wires as is. Each gate $v^{(1)}_{i,n+j}$ for $j\in [n]$ receives a single wire from $v^{(0)}_{i,j}$ and outputs it, and each gate $v^{(1)}_{i,n}$ for $j\in [n]$ receives all input wires and computes the messages that node $i$ sends in the first round of $\mathcal{A}$. Each gate $v^{(2)}_{i,j}$ for $j\in [n]$ receives one wire, which is from $v^{(1)}_{j,i}$ and thus its content is the message that node $i$ receives from node $j$ in the first round of $\mathcal{A}$. Each gate $v^{(2)}_{i,n+j}$ for $j\in [n]$ receives one wire which is from $v^{(1)}_{i,j}$, and thus holds the $j$-th input of node $i$. Finally, each of the gates $v^{(3)}_{i,j}$ for $j\in [n]$ receives as input a wire from all gates $v^{(2)}_{i,k}$ for $k\in [2n]$ and thus it correctly computes the outputs of node $i$ if $T=1$ or its outgoing message to node $j$ if $T>1$, and the gates $v^{(3)}_{i,n+j}$ for $j\in [2n]$ become copies of the gates $v^{(2)}_{i,j}$ for $j\in [2n]$.

\textbf{Induction Step:}
We assume the induction hypothesis holds up to layer $2\ell'-1$ and prove it for layers $2\ell'$ and $2\ell'+1$. 

By the induction hypothesis, the output wires of $v^{(2\ell'-1)}_{i,j}$ for $i,j\in [n]$ holds the message that node $i$ sends to node $j$ in round $\ell'$ of $\mathcal{A}$ and thus, by construction, the input to $v^{(2\ell')}_{i,j}$ for $i,j\in [n]$ is the message that node $i$ receives from node $j$ in that round. Further, by the induction hypothesis, the output wires of $v^{(2\ell'-1)}_{i,j}$ for $i,j\in [n,(\ell'+1)n]$ hold the inputs to node $i$ and the messages it received up to round $\ell'-1$. Thus, by construction, these are the input and output wires of the gates $v^{(2\ell')}_{i,j}$ for $i,j\in [n,(\ell'+1)n]$. This proves the induction step for layer $2\ell'$. 

For layer $2\ell'+1$, by construction, each gate $v^{(2\ell'+1)}_{i,j}$ for $i,j\in [n]$ receives as input wires all outputs of gates $v^{(2\ell')}_{i,k}$, for $k\in [(\ell'+1)n]$. Since these are all inputs and messages that node $i$ receives, this gate correctly computes the message sent by node $i$ to node $j$ in round $\ell'+1$ of $\mathcal{A}$. Finally, by construction, gate $v^{(2\ell'+1)}_{i,n+j}$ for $j\in [n,(\ell'+1)n]$ gets and outputs the wire from gate $v^{(2\ell')}_{i,j}$ and thus corresponds to the input of node $i$ and all messages it received in rounds 1 to $\ell'$, as needed.
\end{proof}

\section{Layered Circuits as $\bw$-\fclique Algorithms}
\label{sec:circuitTOfclique}

We now prove our main result about how to compute a layered circuit in the $\bw$-\fclique model.

\ThmCiruitToFclique*

An immediate corollary of Theorem \ref{thm:circ_to_clique_coded} is that circuits with small \emph{local parallel partitions} are optimal if the output size is not huge, in the sense that our transformation incurs the smallest overhead for them when implementing them in the \fclique model, compared to what we can say about their \clique implementation, which is $O(n^{\de+\f})$ rounds (see Lemma~\ref{lem:circ_to_clique}). 

\begin{corollary}
\label{corollary:localIsOptimal}
    Let $\mathcal{C}$ be a circuit of depth $n^\de$ with alphabet size $|\Sigma| = 
    b \log(n)$
    that has an $n^{\f}$-local parallel partition (i.e., with computation locality $n^{\f}$ and communication locality $n^{\f}$). Let $\mu$ be such that $O(n^{1+\mu})$ bounds the max size of the output per part. 
    For every constant $\bw$, there is an algorithm $\mathcal{A}$ in the $\bw$-\fclique 
    that computes the function $f_\mathcal{C}$ with $O(1)$  quiet rounds, a (non-quiet) round complexity of $O(\bw^2 (n^{\de+\f}+n^{\mu})\log{n})$, and decodability complexity of $O(n^{\mu})$. 
\end{corollary}

To prove the theorem, we make use of the following codes.
\begin{definition}[\textbf{Error-correcting codes}]
    \label{def:codes}
An $[N,K,d]_q$ code is a mapping $enc:(\mathbb{F}_q)^K\rightarrow (\mathbb{F}_q)^N$, such that the Hamming distance between any two codewords $s_1,s_2 \in Image(enc)\subseteq (\mathbb{F}_q)^N$ is at least $d$. By its definition, an $[N,K,d]_q$ code can correct up to $d-1$ erasures, that is, given a string $s'\in (\mathbb{F}_q)^N$, there is at most one codeword $s\in Image(enc)\subseteq (\mathbb{F}_q)^N$ which equals $s'$ up to at most $d-1$ erasures of symbols in $(\mathbb{F}_q)^N$.
\end{definition}

\begin{lemma}
\label{lemma:code}
Let $n$ and $c$ be some integers.
Let $p$ be a prime number and let $k$ be an integer such that the prime power $q=p^k$ satisfies
$\log q=\log(p^k) = \Theta(\bw b \log ( n))= \Theta(\bw \log |\Sigma|)$,
then there exists an $[N,K,d]_q := [n,n/c,\frac{c-1}{c}n+1]_q$ code. 
\end{lemma}

\begin{proof}
    The existence of a code with these parameters ($d=N-K+1$) is a classical result in coding theory; in particular, since $q \geq n $ (\emph{i.e.,} since the field size is bigger than the length of the code) we have that the classical construction of Reed-Solomon codes suffices. See Section~6.8 of \cite{vanlint98} or Section~5.2 of \cite{Huffman_Pless_2003}. 
\end{proof}

\begin{remark}
Note that if we encode $n b\log{n}$ bits of information then the length of a codeword is $c  n b\log{n}$. However, we consider the codeword as having $n$ symbols, each one held by a different node in the \fclique model. Since faults are in terms of nodes,  we lose exactly $c b\log{n}$ bits of information per fault, and so when we consider the codeword as having $n$ symbols, the size of each symbol in the alphabet becomes $cb\log{n}$.
\end{remark}

\begin{proof}[Proof of Theorem~\ref{thm:circ_to_clique_coded}]
~\\\textbf{Construction of the algorithm:}
We construct $\mathcal{A}$ as follows:    
    \begin{enumerate}
  
        \item We associate node $w$ with the gates in part $P_{0,w}$. In the first $O(1)$ quiet rounds, the nodes shuffle their data so that each node $w$ holds the inputs to its gates. This is possible by Lenzen's routing scheme\cite{Lenzen13}.

        \item In the subsequent $c=O(1)$ quiet rounds, each node $w$ encodes all of its input and sends a coded piece to each of the other nodes, as follows:
        denote by {$\g (\_ )$} multiplication on the right by the $ K \times n$ generator matrix corresponding to the code given by Lemma~\ref{lemma:code}, and define the codeword $\widetilde P_{0,w} = \g (P_{0,w}) \in \mathbb{F}^n_q$. 
        Note that the total number of bits in a codeword is at most $\bw n|\Sigma|$, since the size of the input of each node is $n|\Sigma|$ and the first layer is the identity function by Definition~\ref{def:lay_circ}.  
        We refer to the action of encoding the output wires of a part and splitting the pieces of the codeword to all nodes as \emph{checkpointing} this information.
        Thus, at this point, all nodes have checkpointed the 0-th layer of the circuit.

    \item We now recursively describe the method by which the nodes use the checkpointed values 
    $\widetilde{{P}}_{\ell ,w}^{(k)}$ of a layer $\ell$ to compute the values in ${P}_{\ell +1 ,w}$ and then checkpoint them (but notice that from this point on, we are not guaranteed to have quiet rounds).
    The method of collecting coded checkpoints can be thought of as filling out ``bingo cards'', which represent the tasks of possibly faulty \emph{virtual} nodes; 
    in particular, 
    we now have that the data of node $w$ lives ``up in the cloud'' and obtaining data from $w$ is replaced with obtaining coded pieces 
    from the entire network.

   \begin{enumerate}[label=\roman*.]
        \item (\textbf{Communication}) 
        Assume that the nodes have checkpointed all of the data from the previous round so that every node has a coefficient of the codeword $\widetilde{P}_{\ell, v}^{(k)} := \g( P_{\ell,v}^{(k)}) \in \mathbb{F}_{q}^{n}$, where $k\in [hn^{\zeta}]$ (so there are $hn^\zeta$ codewords of length $n$). The reason that there are $hn^\zeta$ such codewords is because $n^\zeta$ is the computation locality (see Definition \ref{def:locality}).
                In particular, node $u$ has the value
            $( \widetilde{P}_{\ell ,v}^{(k)})_u \in \mathbb{F}_{q}^{n}$.
            
                Each node $w$ collects the pieces 
                \begin{equation}
                \widetilde{ \mathrm{bin}}({P}_{\ell+1,w}) := \{ \g(P_{\ell,v}^{(k')}) \mid P_{\ell,v}^{(k')} \in \mathrm{bin}({P}_{\ell+1 ,w}) \} 
                \end{equation}
that they need for computing the gates contained in their part ${P}_{\ell+1,w}$, where $k'\in [h'n^{\xi}]$ because $n^{\xi}$ is the communication locality (see Definition \ref{def:locality}). In particular, $w$ collects ${P}_{\ell,v_1}^{(k_1)}$ in a first set of $c$ rounds, ${P}_{\ell,v_2}^{(k_2)}$ in a second set of $c$ rounds, ..., and ${P}_{\ell,v_{h'n^\xi}}^{(k_{h'n^\xi})}$ in a $h'n^\xi$-\emph{th} set of $c$ rounds, where $h'$ is the constant in the definition of communication locality. Therefore, this takes $O(cn^\xi)$ rounds.

                \item (\textbf{Computation}) If for every part $P_{\ell,w}$ 
                in this layer, all output wires go to the part $P_{\ell+1,w}$ in the next layer (and not to any $P_{\ell+1,v}$ for some $v\neq w$), we call this layer a \emph{computation layer}. Otherwise, we call the layer a \emph{communication layer}.
                
                If $\ell$ is a computation layer, the node $w$ performs the computation corresponding to the gates in $P_{\ell+1,w}$ locally.
                The node $w$ continues to do these computations locally until there is a layer $\tau := \ell +t $ for some $t$, in which either there is an output wire that goes from a part $P_{\tau,w}$ for some $w$ into a part $P_{\tau+1,v}$ for some $v\neq w$ (that is, $\tau$ is a communication layer), or $\tau$ is the output layer. 
                The round complexity of this step is 0, because these are all local computations.

                \item (\textbf{Checkpointing}) 
                Upon arriving at some communication layer $\tau$, each node $w$ checkpoints and sends each node $u$ the value $( \widetilde{P}_{\tau,w}^{(k)})_u:= (\g (P_{\tau,w}^{(k)}))_u \in \mathbb{F}_{q}^{n}$ for $k\in [h n^\zeta]$ (since $n^\zeta$ is the computation locality).
Checkpointing a code with these parameters requires $O(c)$ rounds, since we have $n$ symbols of size $\log{q}=O(cb\log{n})$ bits. Thus, this completes in $O(cn^\zeta)$ rounds. 

~\\We refer to a non-faulty execution of steps (i)-(iii) as an \emph{epoch}. We have that the complexity of these steps is $O(c(n^{\zeta}+n^{\xi}))$.\\

                \item (\textbf{Repeat/Fill out Missing Bingo Cards}) If all of the parts in a layer have been checkpointed, then the nodes move on to the next epoch; otherwise, they divide the work of the failed nodes and simulate the computations of the gates in their parts until all are checkpointed. In particular, if $F$ nodes failed so far, and of those $F$, there are $F^{*}$ nodes whose parts are yet to have been checkpointed, the $n  - F$ non-faulty nodes evenly divide themselves and simulate those $F^{*}$ parts by rerunning steps (i)-(iii). This is done as follows.
                \begin{enumerate}[label=\alph*.]
                \item  If $F^{*} > n- F$, then each of the $n-F$ non-faulty nodes is assigned to one failed node whose part has not been checkpointed yet, and repeats steps (i)-(iii) for that part. 
                    
                \item Otherwise,
                there are $F^{*} \leq n- F$ such parts $P_{\tau,v_1 }, P_{\tau, v_2},...,P_{v_{F^{*}}}$. 
                We batch these into sets of size at most $6c$, and let each batch be simulated by multiple nodes, as follows. We denote $F^{*}_{batch} = \max\{\lfloor{F^{*}/3c}\rfloor, 1\}$, and $\mult = \lfloor{\frac{n-F}{F^{*}_{batch}}}\rfloor$. We denote by $x$ the remainder when dividing $n-F$ by $F^{*}_{batch}$ (thus, $\mult = \frac{n-F-x}{F^{*}_{batch}}$). The non-faulty nodes, $w_0,\dots, w_{n-F}$, are split into $F^{*}_{batch}$ batches:\\
                $$R_{\{v_0,v_1,\dots,v_{3\bw-1}\}} := \{w_0,\dots,w_{\mult-1}\},$$ 
$$R_{\{v_{3\bw},v_{3\bw+1},\dots,v_{6\bw-1}\}} :=\{w_{\mult},\dots,w_{2\mult -1}\}, \dots,$$ 
$$R_{\{v_{\left(F^{*}_{batch}-1\right)3\bw },
v_{\left(F^{*}_{batch}-1\right)3\bw+1},
\dots,v_{F^{*}-1} \}} :=\{w_{(F^{*}-1)\mult },\dots,w_{{n-F-x}}\}$$
Notice that since $F^{*}_{batch}$ is defined using the floor of the respective ratio, the size of the last batch may be larger than $3c$, but it is at most $6c$. Further, in case $F^{*} $ is less than $3\bw$ then there is only one batch, which may also be small, but it is assigned to all non-faulty nodes. 
For each batch, every node in that batch is assigned to each  $v_j$ in the batch, and repeats steps (i)-(iii) for the corresponding part $P_{\tau,v_j}$. 
               \end{enumerate}
                
                We call each iteration of these steps an \emph{attempt}.
                If all of the parts have been checkpointed and $\tau$ is the output layer, then the algorithm halts. Note that additional failures can occur throughout repeating.
                Also, we stress that although the algorithm describes exchanging messages between any two nodes, whenever it attempts to send (receive) a message to (from) a failed node, such a message is not delivered.

                The complexity of this step is equal to the complexity of $O(c(n^{\zeta}+n^{\xi}))$ rounds for steps (i)-(iii), plus an additive overhead of $O(cn^\zeta)$ rounds which corresponds to the 
                rounds needed to decode the information corresponding to the node that is being simulated. In case (b), since the size of a batch is $O(c)$, this introduces an overhead of $O(c)$ over that round complexity, for a total of at most $O(c^2(n^{\zeta}+n^{\xi}))$ rounds per attempt. The number of attempts is bounded in what follows.
                \end{enumerate}
   \end{enumerate}

\textbf{Correctness and complexity of the algorithm despite faults:}
It is immediate from the construction that the algorithm produces the output of the circuit. Further, it is straightforward that the number of quiet rounds is $O(1)$. 

We prove that $R$-decodability holds by induction, with $R=O(c n^{\zeta})=O( n^{\zeta}$ for all epochs except the last, and $R=O(c n^{\mu})=O(n^{\mu})$ for the last.
The base case is straightforward because the computation takes place through quiet rounds: at the end of step (1), each node holds the inputs to its gates. In step (2), each node encodes its input using the code in Lemma \ref{lemma:code}, and thus $R$-decodability holds with $R=O(c)=O(1)$.
For the the inductive hypothesis, assume that epoch $\ell $ has been correctly checkpointed as at most  
$O(n^\zeta)$ codewords of an $[n,\frac{n}{c},\frac{c-1}{c}n+1]_{q}$-code for each node. Then, by the inductive hypothesis, we have that a node can, in $O(n^\zeta)$ rounds, obtain all of the checkpointed data it needs and then perform its local computations, and simulate any failed nodes. Thus, at the end of step 3(iv), which we show below that it indeed checkpoints the computation of all nodes (also faulty nodes), we have that epoch $\ell+1 $ is  correctly checkpointed as at most  
$O(n^\zeta)$ codewords of an $[n,\frac{n}{c},\frac{c-1}{c}n+1]_{q}$-code for each node. Thus, $R$-decodability holds with $R=O(c n^{\zeta})=O( n^{\zeta})$. For the last epoch, we only need to encode the output and therefore we get $R$-decodability with
$R=O(c n^{\mu})=O( n^{\mu})$. 

It remains to bound the number of attempts needed for step 3(iv), which will also prove that indeed it checkpoints the computation of all nodes (also faulty nodes).
For an attempt, recall that we denote by $F$ the number of nodes that failed so far (and thus the number of non-faulty nodes is $n-F$), and by $F^{*}$ the number of parts that are yet to be checkpointed. Let $F'$ be the additional number of nodes that fail throughout this attempt. We consider the two possible cases, depending on how $F^{*}$ relates to $n-F$.

In the first case, recall that if $F^{*} > n-F$ then each of the non-faulty nodes is assigned to one failed node whose part has not been checkpointed and repeats its computation. Since there are at least $n/c$ non-faulty nodes even despite the additional $F'$ newly failed ones, we have that at least $n/c$ additional tasks get checkpointed in this attempt. Thus, such an attempt can happen at most $c$ times before all $n$ parts are checkpointed for this epoch.

In the second case, we have that $F^{*}\leq n-F$. Recall that each part that is not yet checkpointed is now attempted to be checkpointed by a multiplicity of $\mult$ or $\mult-1$ nodes, where $\mult=\left\lceil\frac{n-F}{F^{*}}\right\rceil \geq 1$. There are two possible sub-cases based on how the value of $F'$ relates to the remaining number of allowed failures $\remainF = (\frac{c-1}{c})\cdot n - F$. It holds that either $F' \geq \remainF\cdot \frac{1}{2}$ or $F' < \remainF\cdot \frac{1}{2}$. The former case can happen at most $\log{((\frac{c-1}{c})\cdot n)} = O(\log{n})$ times before $\remainF$ drops to 0. 

Thus, suppose that the latter case happens in an attempt.  
If there is a single batch because $F^{*} < 3c$, then all of the batch is performed by all of these nodes, and hence must succeed because at least one of them (at least $n/c$ of them) is non-faulty.

Otherwise, since $x$ is the remainder of dividing $n-F$ by $F^{*}_{batch}$, we have $x\leq F^{*}_{batch}-1$. Thus, the actual number of nodes that are successful in checkpointing the tasks they are now in charge of is at least $n-F-F'-x$. 
Possibly, every $\mult = \lfloor{(n-F)/F^{*}_{batch}}\rfloor$ of them are performing the same batch, but even in this worst case,  the number of distinct newly checkpointed batches is at least $(n-F-F'-x)/\mult$. Note that with the notation of $x$, we have that $\mult = (n-F-x)/F^{*}_{batch}$. Further, it holds that $F^{*}_{batch} = \max\{\lfloor F^{*} / (3\bw)\rfloor,1\} \leq \frac{n}{3\bw} \leq (n-F)/3$, where the first inequality is since $F^{*} < n$, and the second is since $n/c \leq n-F$ because the number of non-faulty nodes can never go below $n/c$.

We bound this number of newly checkpointed batches $(n-F-F'-x)/\mult$ in terms of the total number of remaining batches $F^{*}_{batch}$, as follows. 
We have that:
\begin{align*}
\frac{n-F-F'-x}{\mult} 
&= (n-F-F'-x)/\frac{n-F-x}{F^{*}_{batch}}\\
&= F^{*}_{batch}\cdot\frac{n-F-F'-x}{n-F-x}
= F^{*}_{batch}\cdot(1-\frac{F'}{n-F-x}).
\end{align*}
Since $F'<F_{remain}/2$, and by plugging $F_{remain} = (\frac{c-1}{c})n-F$, we get:
\begin{align*}
\frac{n-F-F'-x}{\mult} 
&\geq F^{*}_{batch}\cdot(1-\frac{\remainF\cdot \frac{1}{2}}{n-F -x})
= F^{*}_{batch}\cdot(1-\frac{((\frac{c-1}{c})\cdot n -  F)\cdot \frac{1}{2}}{n-F -x}).
\end{align*}
Further simple algebraic manipulations give:
\begin{align*}
\frac{n-F-F'-x}{\mult} 
&\geq F^{*}_{batch}\cdot\left(1- \frac{  n - F  - n/c}{ 2 ( n -F - x)}\right)\\
&= F^{*}_{batch}\cdot\left(1- \frac{  (n - F)  - n/c}{  2(n -F) -2x}\right)
 =F^{*}_{batch}\cdot\left(1- \frac{  1  - \frac{n/c}{n-F}}{  2 -\frac{2x}{n-F} }\right).
\end{align*}
Removing the $\frac{n/c}{n-F}$ term from the nominator only decreases the value of the expression, as well as using $x\leq F^{*}_{batch}-1$. We thus have:
\begin{align*}
\frac{n-F-F'-x}{\mult} 
&\geq F^{*}_{batch}\cdot\left(1- \frac{  1  }{  2 -\frac{2(F^{*}_{batch}-1)}{n-F}  }\right).
\end{align*}
Since $\frac{n-F}{F^{*}_{batch}} \geq 3$, we finally have:
\begin{align*}
\frac{n-F-F'-x}{\mult} 
&\geq F^{*}_{batch}\cdot\left(1- \frac{  1  }{  2 - \frac{2(F^{*}_{batch}-1)}{3F^{*}_{batch}}  }\right)
\geq F^{*}_{batch}\cdot\left(1- \frac{  1  }{  2 -\frac{2}{3}  }\right) = F^{*}_{batch} \cdot (1/4).
\end{align*}
This yields that after this subcase, the number of batches to be checkpointed drops from $F^{*}_{batch}$ to at most 
$(1/4)\cdot F^{*}_{batch}$.
This implies that this can happen at most $O(\log{n})$ times before all $n$ parts are checkpointed. 

    Thus the algorithm has all the parts checkpointed in $O(c+\log{n})$ attempts. Each attempt runs in either $O(c\cdot (n^\zeta+n^{\xi}))$ or $O(c^2\cdot (n^\zeta+n^{\xi}))$ rounds, where the latter number of rounds occurs only in the case of batching, which can happen at most $O(\log{n})$ attempts. Hence, we obtain a total of $O(c^2\cdot(n^\zeta+n^{\xi})\log{n})$ rounds per epoch. 
    In the final epoch, the nodes checkpoint the outputs in $O(c^2  n^{\mu} \log n)$ rounds. This is done only once, and thus incurs only an additive overhead.
    In total, since the depth of the circuit is $n^{\de}$, we get a running time of $O(c^2\cdot (n^{\de+\f + \eta} + n^{\mu}) \log{n})$, where $\eta=\max\{\zeta-\f, \xi-\f\}$, as claimed.
    \end{proof}

\subsection{Computing a Circuit in the Non-Faulty \clique}
\label{appendix:non-faulty}

As promised, we now show that without faults, the circuit can be computed by a \clique algorithm with an $O(n^{\de +\f})$ round complexity.

\begin{lemma}\label{lem:circ_to_clique}
    If $\mathcal{C}$ is a layered circuit of depth $n^\de$ with alphabet size $|\Sigma| = b\log(n)$ that has a $n^\f$-parallel partition, 
    then there is an algorithm $\mathcal{A}$ in the \clique model that computes the function $f^\mathcal{C}$ in $O(n^{\de +\f})$ rounds. 
\end{lemma}
\begin{proof}
    We construct the \clique algorithm $ \mathcal{A}$ similarly to the proof of Theorem~\ref{thm:circ_to_clique_coded}, with the exception that we remove all of the (technically demanding) steps pertaining to fault tolerance (\emph{i.e.,} we don't checkpoint). Formally, we omit steps (iii) and (iv) and slightly modify the communication step (i) since we do not encode/decode, as follows (the computation step (ii) remains the same):

i. (\textbf{Communication}) Assuming that the nodes correctly computed all of the gates in the previous layer. Each node $w$ obtains the values of the gates from layer $\ell$ needed to compute all of the values of its gates in $\ell +1$:
                \begin{equation*}
                {P}_{\ell+1,w} := \left\{  v \in V_ \ell \mid( \exists v' \in {P}_{\ell+1,w})\  (( v , v')  \in E_{\ell, \ell+1 } ) \right\} .
                \end{equation*}
The gates $v \in {P}_{\ell+1,w} $ replace the gates $ v\in {P}_{\ell,v}^{(k)}\in\mathrm{bin}({P}_{\ell+1 ,w}) $ in the proof of Theorem~\ref{thm:circ_to_clique_coded}. 
The values in layer $\ell+1 $ can be computed in $O(n^\f)$ rounds since $w$ needs to receive and send at most $O(n^{1+\f})$ data by the definition of parallel partition (Definition~\ref{def:par_part_graph}).

    ~\\The proof of correctness is a straightforward induction on the layers (as in the proof of Theorem~\ref{thm:circ_to_clique_coded}) and since there are $O(n^\de)$ epochs which finish in $O(n^\f)$ rounds each, we have the claimed round complexity.
\end{proof}

\subsection{Coding is Essential}
\label{appendix:coding-is-essential}

We establish that coding is essential in the $\bw$-\fclique model, in the following sense.

\begin{definition}[\textbf{Hamming distance, Incompressible functions}]
The \emph{Hamming distance} between two strings $s,t \in \Sigma^p $ is the number of indices in $[p]$ in which they differ.
The $\epsilon$-ball $B_{\epsilon}(s)$ around a string $s$ is the set of strings with Hamming distance at most $\epsilon$ from it.

    Let $p$ be a multiple of $n$. 
    Let $s$ be a string in $\Sigma^p $, let $m=\frac{p}{n}$, and denote $s^{(w)} : = s_{wm}...s_{(w+1)m-1}$ for $w\in [n]$.
    We say that a function $f:\Sigma ^{p} \rightarrow \Sigma^ {p} $ is \emph{incompressible} for $w\in [n]$, if it holds that
    \begin{equation*}
       |f(\{t \in B_{m}(s) \mid d(s^{(w)},t^{(w)}) = m \})| \geq |\Sigma|^m.
    \end{equation*}
    That is, the function $f$ takes exactly $|\Sigma|^m$ different values over strings $t$ in the ball $B_{m}(s)$ which differ from $s$ in exactly the $m$ indices of $s^{(w)}$. Informally, changing $s^{(w)}$ alone changes the output of the function.
\end{definition}

\begin{definition}[\textbf{Mega-rounds}]
\label{def:megaRounds}
    We say that $r$ consecutive rounds $t,\dots,t+r-1$ form a \emph{mega-round} if the messages each node sends during these rounds depend only on its state at the beginning of round $t$ (and not on any messages received during these rounds).
\end{definition}

\begin{lemma}\label{lem:must_code}
    Consider a function $f^{\mathcal{A}}$ computed by a clique algorithm $\mathcal{A}$, and assume that it is incompressible for the string that represents the input of a node $w$. Then, if there are $c$ quiet rounds and they form a mega-round, then node $w$ must locally encode all of their data using an MDS (not necessarily linear) code of type $[ n,n/\bw,(\bw -1)n/\bw +1 ]_q$. 
\end{lemma}
\begin{proof}
     We prove that if $w$ does not encode all of its data then, the rest of the network is not able to distinguish between two of $w$'s states, which contradicts the function $f^{\mathcal{A}}$ being incompressible.
Consider the messages that node $w$ sends to all of the nodes in the first $c$ quiet rounds 
as a word $s_w\in \Sigma^{\bw n}$, 
which can be considered a codeword of an $[n,\frac{n}{\bw},d]_q$-code. 
The proof is completed by showing that the distance of this code must be $d = \frac{c - 1 }{c} n +1 $.
Suppose that the distance $d$ of the code is at most $D = \frac{\bw - 1}{\bw} n $. Then there is an input $x_w$ for which the Hamming ball around its corresponding codeword $B_{D}(s_w)$ contains the codeword $s_w'$ of another input $x_w' \neq x_w$, so that $s_w$ cannot be uniquely decoded (or distinguished from $s_w'$) given $D$ failures. 
If the adversary fails $w$ along with $\frac{c-1}{c}n - 1$ other nodes in the first noisy round, then considering $s_w$ as a codeword we have more than $d$ erasures in it. 
This means that there are two states (\emph{i.e.,} inputs) of $w$ that are indistinguishable for the rest of the network, so if the nodes correctly compute $f^\mathcal{A}$, this contradicts the assumption that $f^\mathcal{A}$ is incompressible.  
\end{proof}

\section{Application: Semi-Ring Matrix Multiplication}
\label{appendix:semi-ringMM}

\ThmMatrix*

\begin{proof}
We prove the theorem by showing that there is circuit computing this function that satisfies the hypothesis of Theorem~\ref{thm:circ_to_clique_coded}.  
In order to describe our circuit, we first define some auxiliary notation.  

\textbf{Outer Partition:}
The matrices are split as follows: 
$A$ is split into $n^{1/3} \times n^{1/3} $ block matrices of equal dimension $ n^{2/3} \times  n^{2/3}$; \emph{i.e.,}
\begin{equation*}
    A =: \begin{bmatrix}
        A^{1}_1 & \dots & A^{1}_{n^{1/3} }\\ 
  \vdots & \dots & \vdots\\ 
    A^{n^{1/3}}_1 & \dots & A^{n^{1/3}}_{n^{1/3}} \\ 
    \end{bmatrix} =: \begin{bmatrix}
        A_1 & \dots & A_{n^{1/3} }\\ 
    \end{bmatrix}
\end{equation*}
where the $A^{i}_j$ are $n^{2/3} \times n^{2/3}$ matrices, and similarly $B$ is split into ${n^{1/3}} \times {n^{1/3}}  $ block matrices of equal dimension $n^{2/3} \times n^{2/3}$; \emph{i.e.,}
\begin{equation*}
    B =: \begin{bmatrix}
        B^{1}_1 & \dots & B^{1}_{n^{1/3}} \\ 
  \vdots & \dots & \vdots\\ 
    B^{{n^{1/3}}} _{1} & \dots & B^{{n^{1/3}}}_{n^{1/3}}  \\ 
    \end{bmatrix}
    =: \begin{bmatrix}
        B^{1}\\ 
  \vdots \\ 
    B^{{n^{1/3}}} \\ 
    \end{bmatrix}
\end{equation*}
where the $B^{i}_j$ are $n^{2/3} \times n^{2/3}$ matrices. Elementary linear algebra gives us that the product $C = A \cdot B$ satisfies the equation 
$C^{i}_j = \sum_{k \in [{n^{1/3}}]} (A_kB^{k})^{i}_j$,
which can be computed by a total of $
(\text{\# of $k$ in }A_k,B^k)\cdot (\text{\# of $i$ in }A_k^i)\cdot (\text{\# of $j$ in }B^k_j) = 
{n^{1/3}} \cdot {n^{1/3}} \cdot {n^{1/3}}=n$ block (outer) matrix multiplications; indeed,
by the definition of an outerproduct we have that $(A_kB^{k})^{i}_j = A_k^{i}B^{k}_j$, 
which we use to more efficiently distribute the tasks, and so we have:

\begin{equation}
\label{eq:Cij}
C^{i}_j = \sum_{k \in [{n^{1/3}}]} (A_kB^{k})^{i}_j = \sum_{k \in [{n^{1/3}}]} A_k^{i}B^{k}_j
\end{equation}

\textbf{Inner Partition:} 
We define \emph{sub}-blocks $A^i_j[m]$ of the blocks $A^i_j$, each of size $n^{2/3}\times n^{1/3}$, as follows:
\begin{equation*}
    (A_j^i [m])^k_\ell := (A^i_j)^{k}_{\ell+mn^{1/3}}
\end{equation*}
\emph{i.e.,} $A^i_j[m]$ is implicitly defined by the equation $A^i_j = \begin{bmatrix}
        A^{i}_j[1] & \dots & A^{i}_j[n^{1/3}]
    \end{bmatrix} $.
Similarly, we define subblocks  $B^i_j[m]$ of the blocks $B^i_j$, each of size $n^{1/3}\times n^{2/3}$, analogously to the previous construction:
\begin{equation*}
    (B_j^i [m])^k_\ell := (B^i_j)^{k+mn^{1/3}}_{\ell}
\end{equation*}
\emph{i.e.,} $B^i_j[m]$ is implicitly defined by the following equation
\begin{equation*}
    B^i_j = \begin{bmatrix}
        B^{i}_j[1] \\ 
  \vdots \\ 
     B^{i}_j[n^{1/3}]  \\ 
    \end{bmatrix} 
    .
\end{equation*}
In particular, we have that 
\begin{equation}\label{eq:out_p_cube}
 A^{i}_kB^{k}_j =
 \sum_{m \in [n^{1/3}]} A^{i}_k[m]B^{k}_j[m],
\end{equation}
which by Equation \ref{eq:Cij} further implies that 
\begin{equation*}
C^i_j =  \sum_{k \in [n^{1/3}]} A^{i}_kB^{k}_j =
 \sum_{k,m \in [n^{1/3}]} A^{i}_k[m]B^{k}_j[m].
\end{equation*}
We will also use the fact that
\begin{equation}\label{eq:second_mult_check}
(C^i_j)^\ell_m =  \sum_{k \in [n^{1/3}]}( A^{i}_kB^{k}_j)^\ell_m,
\end{equation}
for the last layer of our circuit. 

\begin{figure}[htbp]
            \centering
            \includegraphics[trim={1cm 5.25cm 2cm 6.5cm},clip,width=1\linewidth]{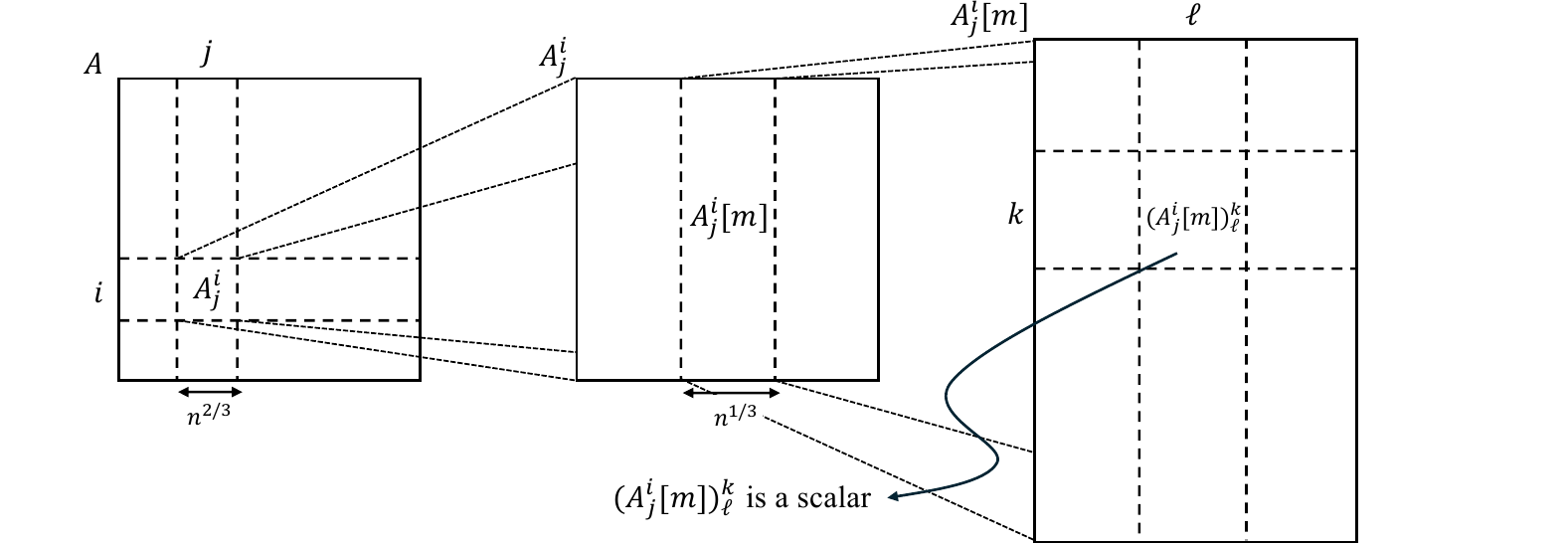}
            \caption{Illustration for the inner and outer partitions, demonstrated on matrix $A$.}
            \label{fig:MatrixMultiplication}
        \end{figure}

~\\\textbf{\underline{Construction of the circuit}}
We construct every layer of the circuit by considering it as having $n$ parts (sets of gates). 
We associate each part $w$ with a 3-tuple of indices 
$(w_1,w_2,w_3) \in [n^{1/3}]^3$.

\textbf{First Layer (Input/Shuffle):} 
We make a part responsible for the entries of its 3-tuple over all of the blocks in the very first shuffling layer; in particular, 
we assume that in the initial shuffle step, the data has been rearranged so that part $w$ now has incoming wires corresponding to
$
(A^{w_1}_{w_2}[w_3])^i_j
$
and
$
(B^{w_1}_{w_2}[w_3])^j_k
$
for all $i,j,k$. 
We can formalize this as the following: 
  the wires of the gates in layer 0 shuffle data to the parts corresponding to $w_0,w_1,w_2$ to part $w$ in layer 1 as given above; \emph{i.e.,} there are $2n$ gates in part $w$, which are:
  \begin{equation*}
P_{0,w} := \left\{(A^{w_1}_{w_2}[w_3])^i_j, (B^{w_1}_{w_2}[w_3])^j_k \mid i,k \in [{n^{2/3}}],j \in [{n^{1/3}}]\right\} .
\end{equation*}

\textbf{Second Layer (Communication):} 
We define the next communication round by the gates 
\begin{equation}\label{eq:cube_layer_two}
  P_{1,w} := \{(A^{w_1}_{w_2}[v])^i_j, (B^{w_2}_{w_3}[v])^j_k \mid v \in [n^{1/3}], i,k \in [{n^{2/3}}],j \in [{n^{1/3}}]\};
\end{equation}
Each gate has a single input wire from the prior layer. 
There are at most $n^{1/3}$ parts whose indices differ from $w = (w_1,w_2,w_3)$ only on $v$ (\emph{i.e.,} for parts $w' = (w_1,w_2,v)$) and they are each connected by 
exactly $n$ {wires} for a total of $n^{4/3}$ wires; therefore, the partition satisfies Definition~\ref{def:lpp}, with:
\begin{itemize}
    \item (Communication Locality) Fix a part $w = (w_1, w_2, w_3)$ and consider its gates $(A^{w_1}_{w_2}[v])^i_j$ for all $v\in[n^{1/3}], i\in[n^{2/3}], j\in[n^{1/3}]$. There are exactly $n^{1/3}$ parts $w' = (w_1,w_2,v)$, one for each such $v$, which have wires into part $w$ in this layer 
    (and each has precisely $n$ wires to part $w$ in this layer, to its gates $(A^{w_1}_{w_2}[v])^i_j$ for the respective $v$).   
    For $(B^{w_2}_{w_3}[v])^j_k$ the count is similar. This gives a communication locality of $O(n^{1/3})$.
    \item (Computation Locality) We have that every part $P_{0,(w_1,w_2,v)}$ can be  partitioned into exactly 2 sets 
   $P_{0,(w_1,w_2,v)}^{(1)}:=\{(A^{w_1}_{w_2}[v])^i_j \mid i \in [{n^{2/3}}],j \in [{n^{1/3}}]\}$  and $P_{0,(w_1,w_2,v)}^{(2)}:=\{ (B^{w_2}_{w_3}[v])^j_k \mid  k \in [{n^{2/3}}],j \in [{n^{1/3}}]\}$. It is straightforward to see that this satisfies computation locality with $n^\zeta = 1$ (and $h=2$). 
\end{itemize}

\textbf{Third Layer (Computation):} 
In the third layer, part $w$ is constructed so that its gates locally compute $A^{w_1}_{w_2}B^{w_2}_{w_3}$. By saying that a gate in part $w$ locally computes a function, we mean that all of its incoming wires are from part $w$ of the previous layer.
That is, we define the $n^{4/3}$ gates:
\begin{equation*}
    P_{2,w} := A^{w_1}_{w_2}B^{w_2}_{w_3} =
 \sum_{v \in [n^{1/3}]} A^{w_1}_{w_2}[v]B^{w_2}_{w_3}[v],
\end{equation*}
one for each element of the $n^{2/3} \times n^{2/3}$ matrix.
By Equation~\ref{eq:out_p_cube} and the definition of the gates in Equation~\ref{eq:cube_layer_two}, we get that indeed these are the values that are computed by the gates in this layer.

\textbf{Fourth Layer (Communication):} 
We define the next communication round by the sets 
\begin{equation}\label{eq:cube_layer_four}
P_{3,w} := \{(A^{w_1}_{v}  B^{v}_{w_3})^{w_2n^{1/3}+i}_{j} \mid  v,i \in [{n^{1/3}}],j \in [{n^{2/3}}]\};
\end{equation}

\begin{itemize}
    \item (Communication Locality) For $(A^{w_1}_{v}  B^{v}_{w_3})^{w_2n^{1/3}+i}_{j} $ there are at most $n^{1/3}$ parts $w' = (w_1,v,w_3)$ that have the information corresponding to $A^{w_1}_{v}  B^{v}_{w_3}$ and furthermore there are precisely $n$ input wires for each one, since $i,j$ have in total a range of $n$.  
    \item (Computation Locality) The parts in the third layer that have outputs wires are all of size $n^{4/3}$ gates, which can each be partitioned perfectly into $n^{1/3}$ sets
    \begin{equation*}
P_{2,(w_1,v,w_3)}^{(w_1,w_2,w_3)} := \{(A^{w_1}_{v}  B^{v}_{w_3})^{w_2n^{1/3}+i}_{j} \mid  v,i \in [{n^{1/3}}],j \in [{n^{2/3}}]\}
    \end{equation*}
   (one for each of the $n^{1/3}$ choices for $w_2$). Each set is of size $n$ and these sets, $P_{2,(w_1,v,w_3)}^{(w_1,w_2,w_3)}$, are all of the subparts corresponding to $(w_1,v,w_3)$. 
   Thus, we have that $n^\zeta = n^{1/3}$.
\end{itemize}

\textbf{Fifth Layer (Computation/Output):} 
Now we can have gates of part $w$ in the fifth layer that can locally (with incoming wires from the fourth layer only from part $w$) compute 
\begin{equation*}
    P_{4,w} := \{(C^{w_1}_{w_2})^{w_3n^{1/3}+i}_j,  \mid i \in [{n^{1/3}}], j \in [{n^{2/3}}]\} =  
    \left \{\sum_{v \in [n^{1/3}]}(A^{w_1}_{v}B^{v}_{w_3})^{w_3n^{1/3}+i}_j \mid i \in [{n^{1/3}}], j \in [{n^{2/3}}] \right\}.
\end{equation*}
By Equation~\ref{eq:second_mult_check} and the definition of the gates in Equation~\ref{eq:cube_layer_four}, we get that indeed these are the values that are computed by the gates in this layer.

~\\This completes the proof that the circuit correctly computes the product of the matrices. The circuit has a constant depth, and an $O(n^{1/3})$-parallel partition. By Theorem~\ref{thm:circ_to_clique_coded}, it can be computed in the $\bw$-\fclique model within $O( c^2n^{1/3}\log{n})$ rounds, with $O(1)$ quiet rounds and $O(1)$-decodability.
\end{proof}

\section{Leveraging a Sublinear Bound on the Number of Faults}
\label{sec:sublinear}
If we have a stronger guarantee of at most $o(n^\chi)$ faults for some $\chi <1$, we can improve the construction. 
Concretely, the lack of locality incurs less overhead, because the worst-case scenario of having to decode \emph{useless} pieces of information now decreases from $n$ to $n^{\chi}$. 
With this in mind, we use a generalized definition of the model, which we refer to as the $(\chi,c)$-\fclique model, in which the adversary may fail $\frac{(c-1)n^{\chi}}{c}$ nodes.     The trick is that while we previously had that the nodes formed a codeword of length $n$, now we make sure that they are partitioned into $n^{1-\chi}$ disjoint codewords of length $n^{\chi}$. 

We  modify the definitions of the locality parameters of a circuit, as follows (the modifications appear in red).

\begin{definition}[\textbf{$\chi$-Computation locality and  $\chi$-Communication locality}]\label{def:alpha_locality}
Let $\mathcal{C} = (V,E) = (\cup_iV_i,E)$ be a layered circuit of depth $n^{\de}$ and an $n^\f$-{parallel partition} with respect to a refinement $\mathcal{P}$ of $V$. Let $P_{i,j}$ denote the $j$-th part of $V_i$ in $\mathcal{P}$.

We say that $\mathcal{C}$ has $\chi$-computation locality $n^\zeta$ and $\chi$-communication locality $n^{\xi}$ 
if there is a constant $h$ such that for all $P_{i,w}$ 
there exists a further (not necessarily disjoint) subdivision 
$\color{red}{P_{i,w} =: \bigcup_{j \in [hn^{\zeta+1-\chi}]}P^{(j)}_{i,w}}$
such that 
$\color{red}{|P^{(j)}_{i,w}| = n^{\chi}}$.
For each $u$, 
the number of pairs $j,w$ such that $w\neq u$ for which $P^{(j)}_{i,w}$ has a wire into $P_{i+1,u}$ is at most 
$\color{red}{h'n^{\xi+1-\chi}}$ for some constant $h'$. We denote by $\bin(P_{i+1,u})$ the parts corresponding to those pairs, that is, the parts $P_{i,w}^{(j)}$ (where $w\neq u$)
in $V_{i}$ 
that consist of at least one gate that has a wire into a gate in the part $P_{i+1,u}$.
\end{definition}

\begin{definition}[\textbf{Local parallel $\chi$-partition}] \label{def:lpp_alpha}
    We say that a layered circuit $\mathcal{C} = (V,E) = (\cup_iV_i,E)$ of depth $n^{\de}$ has an $n^\f$-{local parallel $\chi$-partition}, 
    if it has an $n^\f$-parallel partition 
    for a refinement $V = \cup_{P \in \mathcal{P}}P$
    with a 
    $\chi$-computation locality of $O(n^{\f})$ and a $\chi$-communication locality of $O(n^\f)$.
\end{definition}

Note that having a sublinear bound of $O(n^{\chi})$ faults for $\chi<1$ may allow a circuit to have smaller $\chi$-computation/communication localities compared to the case $\chi=1$. We leverage this to obtain the following.

\begin{theorem}
\label{thm:circ_to_clique_little_o}
[Computing a layered circuit by an $(\chi,c)$-\fclique algorithm]
    Let $\mathcal{C}$ be a circuit of depth $n^\de$ with alphabet size $|\Sigma| = 
    b \log(n)$
    that has an $n^{\f}$-parallel $\chi$-partition with $\chi$-computation locality $n^{\zeta}$ and $\chi$-communication locality $n^{\xi}$. Let $\mu$ be such that $O(n^{1+\mu})$ bounds the max size of the output per part.
    For every constant $\bw$, there is an algorithm $\mathcal{A}$ in the $(\chi,\bw)$-\fclique 
    that computes the function $f_\mathcal{C}$ with $O(1)$ quiet rounds, a (non-quiet) round complexity of $O(c^2(n^{\de+\f+\eta}+n^{\mu})\log{n})$, where $\eta=\max\{\zeta-\f, \xi-\f\}$, and decodability complexity of $R=O( n^{\mu})$.
\end{theorem}

\begin{proof}[Proof of Theorem~\ref{thm:circ_to_clique_little_o}]
    We construct the $(\chi,c)$-\fclique algorithm $ \mathcal{A}$ similarly to the one in the proof of Theorem~\ref{thm:circ_to_clique_coded}, with the only difference being that we replace $O(n)$ with $O(n^\chi)$ where necessary:

    \begin{enumerate}
    \item The first $O(1)$ shuffling rounds remain the same.

    \item In the subsequent $c=O(1)$ quiet rounds, each node partitions its input into $n^{1-\chi}$ parts, each of size $n^{\chi}$. For each part, it forms an MDS code with parameters $[n^{\chi}, \frac{n^{\chi}}{c},\frac{c-1}{c} n^{\chi}+1]$.
    It sends one codeword to each of $n^{1-\chi}$ subsets of size $n^{\chi}$ of a predetermined partition of the nodes which they all use. 

        In other words, while we previously had that the nodes formed only one codeword of length $n$, now they are partitioned into $n^{1-\chi}$ disjoint codewords of length $n^{\chi}$. 

    \item The recursive steps remain almost the same, with the exception that the codewords now have parameters $[n^{\chi}, \frac{n^{\chi}}{c},\frac{c-1}{c} n^{\chi}+1]$ and that the $n$ nodes do not hold one codeword but instead are partitioned into holding $n^{1-\chi}$ disjoint codewords of length $n^{\chi}$. In particular, the only thing that changes is that $( \widetilde{P}_{\ell ,v}^{(k)})_u$ is now of size $n^\chi$ and $u$ also takes the more restricted range of $n^\chi$, since this is the length of the codewords. 

\end{enumerate}

    ~\\The proof of correctness remains the same up to accordingly replacing $n$ with $n^\chi$ in appropriate places. 
    \emph{E.g.,} for the the inductive hypothesis, we assume that epoch $\ell $ has been correctly checkpointed as at most  
$O(n^{\zeta+1-\chi})$ codewords of an $[n^\chi,\frac{n^\chi}{c},\frac{c-1}{c}n^\chi+1]_{q}$-code for each node.
In particular, although there are $O(n^{\zeta+1-\chi})$ codewords, 
they can be collected concurrently in ``batches'' of size $n^{1-\chi}$ and, likewise, collecting the data from any of these ``batches'' cannot fail since the adversary is constrained to fail at most $\frac{c-1}{c}n^{\chi}$ nodes.
\end{proof}

\section{Application: Fast (Ring) Matrix Multiplication }
\label{sec:fastMM}
We prove here that matrix multiplication over a ring can be computed by a \fclique algorithm with overhead proportional to its fault tolerance.

\ThmFastMatrix*
\begin{proof}
We prove the theorem by showing that there is circuit computing this function that satisfies the hypothesis of Theorem~\ref{thm:circ_to_clique_little_o}.  
In order to describe our circuit, we first define some auxiliary notation.  

Let $(\alpha , \beta , \gamma) $
$\in R ^{ m^\sigma \times m\times m  } \times R ^{m^\sigma \times m\times m  }\times R ^{ m \times m \times m^\sigma   }$ be a fast matrix multiplication tensor of rank $O(m^\sigma)$; \emph{i.e,} a tensor that given as input two $m\times m$ matrices $X,Y$ over some ring, computes their product as follows: 
\begin{align}\label{eq:ten_def}
\begin{split}
    \widehat{X}_k & = \sum_{i,j \in [m]}\alpha^k_{i,j}X^i_j
    \\ 
    \widehat{Y}_k & = \sum_{i,j \in [m]}\beta^k_{i,j}Y^i_j
    \\ 
    (XY)^i_j &= \sum_{k \in [m^\sigma]} \gamma^{i,j}_{k} \widehat{X}_k \widehat{Y}_k .
\end{split}
\end{align}

\textbf{Outer Partition:}
The matrices are split as follows: 
$A$ is split into $n^{1/\sigma} \times n^{1/\sigma} $ block matrices of equal dimension $ n^{1-1/\sigma} \times  n^{1-1/\sigma}$; \emph{i.e.,}
\begin{equation*}
    A =: \begin{bmatrix}
        A^{1}_1 & \dots & A^{1}_{n^{1/\sigma} }\\ 
  \vdots & \dots & \vdots\\ 
    A^{n^{1/\sigma}}_1 & \dots & A^{n^{1/\sigma}}_{n^{1/\sigma}} \\ 
    \end{bmatrix} 
\end{equation*}
where the $A^{i}_j$ are $n^{1-1/\sigma} \times n^{1-1/\sigma}$ matrices, and similarly $B$ is split into ${n^{1/\sigma}} \times {n^{1/\sigma}}  $ block matrices of equal dimension $n^{1-1/\sigma} \times n^{1-1/\sigma}$; \emph{i.e.,}
\begin{equation*}
    B =: \begin{bmatrix}
        B^{0}_1 & \dots & B^{1}_{n^{1/\sigma}} \\ 
  \vdots & \dots & \vdots\\ 
    B^{{n^{1/\sigma}}} _{1} & \dots & B^{{n^{1/\sigma}}}_{n^{1/\sigma}}  \\ 
    \end{bmatrix}
\end{equation*}
where the $B^{i}_j$ are $n^{1-1/\sigma} \times n^{1-1/\sigma}$ matrices.

Setting $m=n^{1/\sigma}$ in Equation~\ref{eq:ten_def}, we get that 
\begin{equation}
\label{eq:Cij_fast}
\begin{split}
 \widehat{A}_k & = \sum_{i,j \in [n^{1/\sigma}]}\alpha^k_{i,j}A^i_j
    \\ 
    \widehat{B}_k & = \sum_{i,j \in [n^{1/\sigma}]}\beta^k_{i,j}B^i_j.
    \\ 
    C^{i}_j & = \sum_{k \in [n]} \gamma^{i,j}_{k} \widehat{A}_k \widehat{B}_k .
\\
\end{split}
\end{equation}

\textbf{Inner Partition:} 
We define \emph{sub}-blocks $A^i_j[m,p]$ of the blocks $A^i_j$, each of size $n^{1/2-1/\sigma}\times n^{1/2-1/\sigma}$, as follows:
\begin{equation*}
    (A_j^i [m,p])^k_\ell := (A^i_j)^{k+mn^{1/2-1/\sigma}}_{\ell+pn^{1/2-1/\sigma}}
\end{equation*}
\emph{i.e.,} $A^i_j[m]$ is implicitly defined by the equation 
$$
A^i_j = \begin{bmatrix}
        A^{i}_j[1,1] & \dots & A^{i}_j[1,n^{1/2}]\\
        \vdots & \ddots & \vdots\\
        A^{i}_j[n^{1/2},1] & \dots & A^{i}_j[n^{1/2},n^{1/2}]\\
    \end{bmatrix} .
    $$
Similarly, we define \emph{sub}-blocks $B^i_j[m,p]$ of the blocks $B^i_j$, each of size $n^{1/2-1/\sigma}\times n^{1/2-1/\sigma}$, as follows:
\begin{equation*}
    (B_j^i [m,p])^k_\ell := (B^i_j)^{k+mn^{1/2-1/\sigma}}_{\ell+pn^{1/2-1/\sigma}}
\end{equation*}
\emph{i.e.,} $B^i_j[m]$ is implicitly defined by the equation 
$$
B^i_j = \begin{bmatrix}
        B^{i}_j[1,1] & \dots & B^{i}_j[1,n^{1/2}]\\
        \vdots & \ddots & \vdots\\
        B^{i}_j[n^{1/2},1] & \dots & B^{i}_j[n^{1/2},n^{1/2}]\\
    \end{bmatrix} .
    $$
\textbf{Inner Partition of Auxiliary Matrices:} 
We likewise form an (identical) inner partition of the auxiliary matrices of the form 
$$
\widehat A_k= \begin{bmatrix}
        \widehat  A_k[1,1] & \dots &
 \widehat A_k[1,n^{1/2}]\\
        \vdots & \ddots & \vdots\\
        \widehat A_k [n^{1/2},1] & \dots & \widehat A_k [n^{1/2},n^{1/2}]\\
    \end{bmatrix} ,
    $$
$$
\widehat B_k = \begin{bmatrix}
        \widehat B_k [1,1] & \dots & \widehat B_k [1,n^{1/2}]\\
        \vdots & \ddots & \vdots\\
        \widehat B_k [n^{1/2},1] & \dots & \widehat B_k [n^{1/2},n^{1/2}]\\
    \end{bmatrix} ,
    $$
    and
    $$
\widehat A_k\widehat B_k = \begin{bmatrix}
        (\widehat A_k\widehat B_k) [1,1] & \dots & (\widehat A_k\widehat B_k) [1,n^{1/2}]\\
        \vdots & \ddots & \vdots\\
        (\widehat A_k\widehat B_k) [n^{1/2},1] & \dots & (\widehat A_k\widehat B_k) [n^{1/2},n^{1/2}]\\
    \end{bmatrix} .
    $$
    Similarly to Equation~\ref{eq:second_mult_check} we also have that 
\begin{equation}
\label{eq:second_mult_check_fast}
\begin{split}
C^{i}_j [m,p]& = \sum_{k \in [n]} \gamma^{i,j}_{k}\left( \widehat{A}_k \widehat{B}_k\right)[m,p] .
\\
\end{split}
\end{equation}

~\\\textbf{\underline{Construction of the circuit}}
We construct every layer of the circuit by considering it as having $n$ parts (sets of gates). 
We associate each part $w$ with a 2-tuple of indices 
$(w_0,w_1) \in [n^{1/2}]^2$.

\textbf{First Layer (Input/Shuffle):} 
If $\chi < 1-2\sigma$, we simply design the circuit as if the adversary could fail with parameter $\chi = 1-2\sigma$ (certainly if we prepare for \emph{more} faults, then the circuit would still be correct). 
Therefore, we construct the circuit and perform the analysis under the assumption that $\chi \geq  1-2/\sigma$. 
We make a part responsible for the entries of its 2-tuple over all of the blocks in the very first shuffling layer; in particular, 
we assume that in the initial shuffle step, the data has been rearranged so that part $w$ now has incoming wires corresponding to
$
(A^{i}_{j}[w_0,w_1])^k_\ell
$
and
$
(B^{i}_{j}[w_0,w_1])^k_\ell
$
for all $i,j,k$. 
We can formalize this as the following: 
  the wires of the gates in layer 0 shuffle data to the parts corresponding to $w_0,w_1$ to part $w$ in layer 1 as given above; \emph{i.e.,} there are $2n$ gates in part $w$, which are:
  \begin{equation*}
P_{0,w} := \left\{((A^{i}_{j}[w_0,w_1])^k_\ell, (B^{i}_{j}[w_0,w_1])^k_\ell \mid i,j \in [{n^{1/\sigma}}],k,\ell \in [{n^{1/2-1/\sigma}}]\right\} .
\end{equation*}

\textbf{Second Layer (Computation):} 
In the second layer, part $w$ is constructed so that its gates locally compute $(\widehat A_k)[{w_0},{w_1}], (\widehat B_k)[{w_0},{w_1}]$. 
That is, we define the $2n^{2-2/\sigma}$ gates:
\begin{equation*}
\begin{split}
    P_{2,w}^{(0,k)} = \widehat{A}_k [w_0,w_1] = \sum_{i,j \in [n^{1/\sigma}]}\alpha^k_{i,j}A^i_j[w_0,w_1]
    \\ 
    P_{2,w}^{(1,k)} = \widehat{B}_k [w_0,w_1] = \sum_{i,j \in [n^{1/\sigma}]}\beta^k_{i,j}B^i_j[w_0,w_1]
\end{split}
\end{equation*}
for each $k \in [n]$ this is a pair of $n^{1/2-1/\sigma}\times n^{1/2-1/\sigma}$ matrices.

\textbf{Third Layer (Communication):} 
We define the next communication round by the sets 
\begin{equation*}\label{eq:fast_layer_three}
\begin{split}
    P_{3,w} = \left \{ \widehat{A}_w [v_0,v_1]
     , \widehat{B}_w [v_0,v_1] \mid v_0,v_1 \in [n^{1/2}]\right\}
\end{split}
\end{equation*}

\begin{itemize}
    \item (Communication Locality) For $ \widehat{A}_w [v_0,v_1], \widehat{B}_w [v_0,v_1]$ there are $n$ parts $v$ that have the information corresponding to $\widehat{A}_w [v_0,v_1], \widehat{B}_w [v_0,v_1]$ and furthermore there are precisely $2n^{2-2/\sigma}$ input wires for each one, since the dimensions of $ \widehat{A}_w [v_0,v_1], \widehat{B}_w [v_0,v_1]$ is $n^{1-2/\sigma}$ and there $n$ pairs $(v_0,v_1)$.  
    Therefore,  $ n = n^{\xi+1-\chi} 
 $ implies that we have an $\chi$-communication locality of $\xi = \chi$. 
    \item (Computation Locality) The parts in the second layer that have output wires are all of size $2n^{2-2/\sigma}$ gates, which can each be partitioned perfectly into $n$ sets
    \begin{equation*}
 P_{3,w}^{(0,v)} = \widehat{A}_w [v_0,v_1]
    , \ \ 
    P_{3,w}^{(1,v)} = \widehat{B}_w [v_0,v_1]
    \end{equation*}
   (one for each of the $n$ choices for $v$). Each set is of size $n^{1-2\sigma}$ and these sets, $ P_{3,w}^{(0,v)}, P_{3,w}^{(1,v)} $, are all of the subparts corresponding to $w$. 
   Thus, $ n = n^{\zeta+1-\chi} $ implies that we have an $\chi$-communication locality of  $\zeta = \chi $.
\end{itemize}

\textbf{Fourth Layer (Computation):} 
In the fourth layer, part $w$ is constructed so that its gates locally compute $\widehat A_wB_w$. 
That is, we define the $n^{2-2/\sigma}$ gates:
\begin{equation*}
\begin{split}
    P_{4,w} = \widehat{A}_w  \widehat{B}_w. 
\end{split}
\end{equation*}

\textbf{Fifth Layer (Communication):} 
We define the next communication round by the sets 
\begin{equation*}\label{eq:fast_layer_five}
\begin{split}
    P_{5,w} = \left\{(\widehat{A}_v \widehat{B}_v )[w_0,w_1]\mid v_0,v_1 \in [n^{1/2}]\right\}
\end{split}
\end{equation*}

\begin{itemize}
    \item (Communication Locality) For $ (\widehat{A}_v \widehat{B}_v) [w_0,w_1]$ there are $n$ parts $v$ that have the information corresponding to $(\widehat{A}_v  \widehat{B}_v) [w_0,w_1]$ and furthermore there are precisely $n^{2-2/\sigma}$ input wires for each one, since the dimensions of $ (\widehat{A}_v \widehat{B}_v) [w_0,w_1]$ is $n^{1-2/\sigma}$ and there are $n$ nodes $v$.  
    Therefore,  $ n = n^{\xi+1-\chi} 
 $ implies that we have an $\chi$-communication locality of $\xi = \chi$. 
    \item (Computation Locality) The parts in the second layer that have output wires are all of size $n^{2-2/\sigma}$ gates, which can each be partitioned perfectly into $n$ sets
    \begin{equation*}
  P_{5,w}^{(v)} = (\widehat{A}_v \widehat{B}_v )[w_0,w_1]
       \end{equation*}
   (one for each of the $n$ choices for $v$). Each set is of size $n^{1-2\sigma}$ and these sets, $ P_{5,w}^{(v)} $, are all of the subparts corresponding to $w$. 
   Thus, $ n = n^{\zeta+1-\chi} $ implies that we have an $\chi$-communication locality of  $\zeta = \chi $.
\end{itemize}

\textbf{Sixth Layer (Computation/Output):} 
In the sixth layer, part $w$ is constructed so that its gates locally compute $ C^{i}_j [w_0,w_1]$ for all $i,j$. 
That is, we define the $n$ gates:
\begin{equation*}
\begin{split}
    P_{6,w} = \left\{ C^{i}_j [w_0,w_1]  \mid i,j \in [n^{1/\sigma}]\right\} =   \left\{ \sum_{k \in [n]} \gamma^{i,j}_{k}\left( \widehat{A}_k \widehat{B}_k\right)[w_0,w_1]
     \mid i,j \in [n^{1/\sigma}]\right\}.
\end{split}
\end{equation*}

~\\This completes the proof that the circuit correctly computes the product of the matrices. The circuit has a constant depth, and an $O(n^{1-2/\sigma})$-parallel $\chi$-partition. By Theorem~\ref{thm:circ_to_clique_coded}, it can be computed in the $(\chi,\bw)$-\fclique model within $O( c^2n^{\chi}\log{n})$ rounds, with $O(1)$ quiet rounds and $O(1)$-decodability.
In the case where $\chi \leq 1-2/\sigma$, then we get that it can be computed in the $(\chi,\bw)$-\fclique model within $O( c^2n^{1- 2/\sigma}\log{n})$ rounds, with $O(1)$ quiet rounds and $O(1)$-decodability (\emph{i.e.,} we have that $\eta = 0$ in the case where $\chi \leq 1-2/\sigma$).
\end{proof}

\paragraph{Acknowledgments:} We thank Orr Fischer, Ran Gelles, and Merav Parter for useful discussions. This work is supported in part by the Israel Science Foundation (grant 529/23).

\bibliographystyle{alpha}
\bibliography{main.bib}

\end{document}